\def\draft{}

\def\sidebysidefigures{}

\def\arXiv{}

\ifx\draft\undefined
	\let\sidebysidefigures=\undefined
	\documentclass[journal, twocolumn, twoside]{IEEEtran}
\else
	\documentclass[journal, draftclsnofoot, 12pt, onecolumn, twoside]{IEEEtran}
\fi

\usepackage{amsmath}
\usepackage{amssymb}
\usepackage{amsthm}
\usepackage{dsfont}\let\mathbb\mathds
\usepackage{wasysym}
\usepackage{graphicx}\graphicspath{{./figs/}}
\usepackage[caption=false]{subfig}
\usepackage{cite}
\usepackage{booktabs}
\usepackage{microtype}
\usepackage[hyphens]{url}

\ifx\arXiv\undefined\else

\usepackage{color}
\definecolor{darkblue}{rgb}{0.0, 0.0, 0.55}
\usepackage{hyperref}
\hypersetup{
    pdfstartview=Fit,
    linktoc=all,
    colorlinks=true,
    linkcolor=darkblue,
    citecolor=darkblue,
    filecolor=darkblue,
    urlcolor=darkblue,
    bookmarksopen=false,
    bookmarksopenlevel=2,
    bookmarksdepth=3,
} 
\usepackage[all]{hypcap}

\let\sidebysidefigures=\undefined

\fi

\theoremstyle{plain}
\newtheorem{lemma}{Lemma}
\newtheorem{theorem}{Theorem}
\newtheorem{proposition}{Proposition}
\newtheorem{corollary}{Corollary}[theorem]
\theoremstyle{definition}
\newtheorem{definition}{Definition}

\newtheorem{assumption}{Assumption}
\theoremstyle{remark}
\newtheorem{remark}{Remark}

\def\nbu{{\mathbf{u}}}
\def\nbv{{\mathbf{v}}}

\def\nbz{{\mathbf{z}}}

\def\nbX{{\mathbf{X}}}

\def\ncalA{{\mathcal{A}}}
\def\ncalB{{\mathcal{B}}}

\def\ncalF{{\mathcal{F}}}

\def\ncalL{{\mathcal{L}}}

\def\ncalN{{\mathcal{N}}}

\def\ncalS{{\mathcal{S}}}

\def\ncalU{{\mathcal{U}}}
\def\ncalV{{\mathcal{V}}}

\def\nbbA{{\mathbb{A}}}

\def\nbbG{{\mathbb{G}}}

\def\nbbN{{\mathbb{N}}}

\def\nrmc{{\rm c}}

\def\nrmn{{\rm n}}

\def\E{\mathbb{E}}
\def\P{\mathbb{P}}
\def\R{\mathbb{R}}
\def\sreq#1{\stackrel{(#1)}{=}}
\def\indicator{{\mathbb{1}}}

\def\d{\mathrm{d}}
\def\T{\top}

\def\th{{^\text{th}}}

\def\pl{\mathtt{P_L}}
\def\plk{\mathtt{P_L^K}}
\def\PPPa{{\Phi}}

\def\pppai{{\lambda}}
\def\PPPb{{\Psi}}

\def\pppbi{{\nu}}

\def\N{{\sigma^2}}
\def\threshold{{\beta}}
\def\SINR{{\sf SINR}}

\def\origin{{o}}

\def\Gammal{{(\ell-1)!}}
\def\GammaK{{(K-1)!}}
\def\da{{\nbu}}
\def\db{{\nbv}}
\def\Alune{{\nbbA_\text{\rm\tiny\leftmoon}}}
\def\seta#1{{\ncalS_{\da}^{[#1]}}}
\def\setb#1{{\ncalS_{\db}^{[#1]}}}
\def\bRL{{\bar{R}_\ell}}
\def\bL{{\sqrt{r^2+d^2-2rd\cos \theta}}}
\def\bLx{{\sqrt{R_{\ell}^2+d^2-2 R_{\ell} d\cos \theta}}}
\def\ptx{x}
\def\dptxtob{\bar{X}}
\def\setdptxtob{{\bar{\nbX}}}
\def\phirange{{\phi_{\text{range}}}}
\def\dK{y}
\def\L{\ncalL}
\def\Lc{\L_{\nrmc}}
\def\Lnc{\L_{\nrmn\nrmc}}

\allowdisplaybreaks
\newlength{\maxfigurewidth}
\newlength{\figurewidth}
\ifx\sidebysidefigures\undefined
\else
\fi

\begin{document}

%!TEX root = paper_2.tex

\title{A Tractable Analysis of the Improvement in Unique Localizability Through Collaboration}

\author{%
Javier Schloemann, Harpreet S. Dhillon, and R. Michael Buehrer%
\ifx\arXiv\undefined\else%
	\thanks{This work has been submitted to the IEEE for possible publication. Copyright may be transferred without notice, after which this version may no longer be accessible.}%
\fi%
\thanks{The authors are with the Mobile and Portable Radio Research Group (MPRG), Wireless@Virginia Tech, Blacksburg, VA, USA. Email: \{javier, hdhillon, buehrer\}@vt.edu. This paper is submitted in part to IEEE GLOBECOM 2015 Workshop on Localization for Indoors, Outdoors, and Emerging Networks (LION), San Diego, CA, USA~\cite{Schloemann2015e}.%
\ifx\draft\undefined\else%
	\hfill Manuscript last updated: \today.%
\fi%
}%
}

% The paper headers
%\markboth{IEEE Transactions on Wireless Communications}{Submitted paper}

% Make the title area
\maketitle
%!TEX root = paper_2.tex

\begin{abstract}

In this paper, we mathematically characterize the improvement in device localizability achieved by allowing collaboration among devices. Depending on the detection sensitivity of the receivers in the devices, it is not unusual for a device to be localized to lack a sufficient number of detectable positioning signals from localized devices to determine its location without ambiguity (i.e., to be uniquely localizable). This occurrence is well-known to be a limiting factor in localization performance, especially in communications systems. In cellular positioning, for example, cellular network designers call this the \emph{hearability problem}. We study the conditions required for unique localizability and use tools from stochastic geometry to derive accurate analytic expressions for the probabilities of meeting these conditions in the noncollaborative and collaborative cases. We consider the scenario without shadowing, the scenario with shadowing and universal frequency reuse, and, finally, the shadowing scenario with random frequency reuse. The results from the latter scenario, which apply particularly to cellular networks, reveal that collaboration between two devices separated by only a short distance yields drastic improvements in both devices' abilities to uniquely determine their positions. The results from this analysis are very promising and motivate delving further into techniques which enhance cellular positioning with small-scale collaborative ranging observations among nearby devices.

\end{abstract}
%!TEX root = paper_2.tex

%!TEX root = paper_2.tex

\section{Introduction}\label{Sec:Introduction}

Determining the locations of devices in mobile ad-hoc networks (MANETs), wireless sensor networks (WSNs), and cellular networks has many important applications. In MANETs, which are useful in disaster recovery, rescue operations, and military communications, location information is used to enable location-aided routing and geodesic packet forwarding~\cite{Karp2000, Ko2000, Blazevic2001, Jain2001, Lee2010, Raji2014}. In WSNs, whose applications include environmental monitoring (e.g., for precision agriculture) and asset tracking in warehouses, not only is location information useful for the self-organization of the network, but in addition, tying locations to the sensor observations is crucial for interpreting the sensed data~\cite{Akyildiz2002, Patwari2005, Gezici2005}. Lastly, in cellular networks, which provide nearly ubiquitous communication capabilities, location information is used to provide subscribers with location-based services in addition to providing public service answering points with potentially life-saving location information during emergency calls~\cite{Sayed2005, Gustafsson2005, FCCE911CFR}.

A seemingly simple solution for providing universal location information in the aforementioned networks is to take advantage of prevalent global navigation satellite systems (GNSS), e.g.,  GPS, GLONASS, and Galileo; however, such systems are not always available or reliable. For example, MANETs and WSNs often consist of large numbers of devices, meaning that it is often not economically viable to equip all devices with GNSS receivers. Instead, it is likely that only a portion of the devices are equipped with GNSS receivers, allowing those devices to locate themselves and then serve as reference points with which the remaining devices communicate in order to then calculate their own locations using some well-established localization technique. On the other hand, GNSS receivers are standard equipment in all new cellular devices; however, these devices are often used indoors where the satellite signals may be too weak to provide reliable location estimates. With the advent of new \emph{indoor} location requirements imposed by the Federal Communications Commission (FCC)~\cite{FCC2015}, it is becoming increasingly imperative for cellular network operators to be able to fall back on accurate terrestrial localization techniques.

Classically, the localization procedure is performed separately at the mobile devices (MDs), each communicating only with a set of already-localized reference devices, which the literature commonly calls \emph{anchors} or \emph{beacons}. In cellular networks, \emph{base stations} (BSs) serve as the reference devices, and this is the term we will use since cellular positioning is a primary focus of this work. Now, the first objective in any location system is to make sure that the device to be located can receive positioning signals from a sufficient number of BSs in order to calculate a position fix. This is far from guaranteed, something to which cellular network designers, who call this the \emph{hearability problem}, will attest~\cite{R1-091912}. Lately, collaboration between MDs has received more and more attention as a means to improve positioning performance, both for MANETs and WSNs~\cite{Savvides2001, Gezici2005, Patwari2005, Alsindi2009a, Shen2010b} as well as, more recently, for traditional cellular networks~\cite{Vaghefi2014a}. The primary benefits provided by collaboration are (i) an increased ability to calculate position fixes~\cite{Wymeersch2009} and (ii) more accurate location estimates~\cite{Schloemann2015b}. The objective of this paper is to study the former benefit; more specifically, we ask: how does enhancing the classical localization procedure with a single collaborative link impact the availability of position fixes?

\subsection{Prior art and motivation}

There is a rich body of literature concerning the study of MANETs and WSNs, both for connectivity as well as positioning. The former is important because without connectivity, there is no communication, ultimately rendering the network ineffective. A significant portion of previous work deals with \emph{percolation-based connectivity}~\cite{Dousse2002, Bettstetter2005, Santi2005, Ren2011}, which studies how system parameters affect the ability to obtain an infinite connected component in a network of randomly-distributed devices. It is standard practice to model device locations according to a homogeneous Poisson point process (PPP) since MANETs and WSNs exhibit no concrete backbone structure and deployments are typically not rigidly planned. From the localization perspective, a substantial portion of work is focused on \emph{collaborative} or \emph{cooperative} localization, whereby MDs gather position-related observations amongst themselves in addition to their observations from BSs and simultaneously estimate their locations. While analysis is known to be hard and the majority of papers eventually resort to simulation-based insights (e.g., see \cite{Rekleitis2002, Ihler2005, Alsindi2006, Wymeersch2008}), the results of~\cite{Schloemann2015b} can be used to show that enhancing a classical localization procedure with collaboration is beneficial to positioning accuracy. Specifically, a device making just a single collaborative connection to a secondary device meeting some minimum BS connectivity condition enjoys a strict reduction in the \emph{Cram\'{e}r-Rao lower bound} (CRLB) of its positioning error. However, one thing a CRLB analysis does not do is take into account errors due to ambiguities (such as the \emph{flip ambiguity}~\cite{Eren2004}) since the Fisher information calculations consider only the \emph{peakedness} of the log-likelihood functions of the observations in the vicinity of the true device locations (even though the log-likelihood may be the same for many other locations). Thus, as a complement to \cite{Schloemann2015b}, the present work explicitly considers the value of collaboration as a means for taking a device from not being able to locate itself to where it's able to locate itself without ambiguity, i.e., to localizability.

Much the same as in MANETs and WSNs, the cellular network literature is very rich with connectivity analyses (in the form of \emph{coverage probabilities}), though in contrast, it is not quite as rich in terms of localization studies. Primarily, this is due to the fact that location information is nowhere near as crucial to cellular networks as it is to, say, WSNs, where location tagging adds meaning to the collected data. With the proliferation of smart phones enabling location-based services as well as increased pressure due to federal regulations such as the FCC E911 mandate~\cite{FCCE911CFR}, however, positioning of cellular devices has recently garnered increasing interest. On the whole, the vast majority of both coverage and positioning studies have relied on simulation-based results. This is largely due to the fact that cellular networks are often modeled using widely-accepted grid-based models which do not lend themselves to tractable analyses. An example which is related to the present work is~\cite{Vaghefi2014a}, which used simulation results to show that collaboration provides an increase in the availability of positioning fixes in LTE networks using the OTDOA handset-based positioning method. Recently, several landmark coverage papers have appeared which model cellular BS deployments according to a homogeneous PPP (e.g.,~\cite{Andrews2011, Dhillon2012}), essentially arguing that the popular grid-based models are themselves highly-idealized and becoming more obsolete as cellular networks deviate from centrally-planned macro-cell networks to networks which include an increasing number of more arbitrary small-cell deployments such as picocells and femtocells. This same approach was taken in~\cite{Schloemann2015c}, which analytically studied the hearability of far away BSs for the purposes of cellular positioning after showing that there is a coupling between hearability and a cellular network operator's ability to meet specific location accuracy requirements (e.g., the FCC E911 requirements). Among other things, the results in~\cite{Schloemann2015c} can be used to determine the probability that a device will be unable to locate itself without ambiguity when only using positioning observations from BSs, i.e., noncollaborative localization. The present work expands on this and studies how much a collaborative link helps improve a device's ability to locate itself without ambiguity.

We consider a classical localization procedure extended with a single collaborative link, resulting in an estimation problem which consists of two unknown device locations which are then estimated simultaneously. While more complex than the original noncollaborative estimation problem, the addition of a collaborative link may allow a MD which is otherwise not able to locate itself using a classical procedure to be able to uniquely determine its location. Understanding exactly how often collaboration helps in this regard is the purpose of this study. In accordance with previous works, BSs are modeled according to a homogeneous PPP. The MDs are modeled similarly, using a second (independent) PPP, which agrees with the uniformly-random modeling of MDs in MANETs, WSNs, 3GPP simulations, as well as the cellular literature~\cite{Novlan2013}.

%Besides providing tractability, we contend that this is a reasonable starting point even for cellular networks since the PPP is the parent point process of many other point processes~\cite{ElSawy2013}, including the determinantal point process which has been shown to yield BS deployments which are geographically true to actual cellular BS deployments~\cite{Li2015}.

\subsection{Contributions}
This paper makes several contributions to the study of collaborative localization. The main contributions of this paper are as follows.\vspace{0.5em}

\noindent \emph{Derivation of unique localizability conditions for two-device collaborative localization}: In Section~\ref{P2:Sec:UniqueLoc}, we employ a graph-theoretic approach to derive the conditions required for a MD's location to be determined without ambiguity for the small-scale (specifically, two device) collaborative localization setup using ranging observations. Reasonable conditions required to guarantee unique localizability using range-difference observations to BSs are also presented.\vspace{0.5em}

\noindent \emph{Analytic expressions for the probability of unique localizability with collaboration}: For the scenario without shadowing, we provide an exact analysis of the probability that a device hearing a certain number of BSs and collaborating with a secondary device hearing a certain number of BSs will go from not being uniquely localizable to being uniquely localizable. We then combine these results with our previous work on hearability and provide accurate approximations for the probability of unique localizability with and without collaboration when using range-based and range-difference-based observations from BSs. The resulting expressions account for network self-interference, something which is often omitted from MANET and WSN studies. Furthermore, we also present accurate approximations of these same probabilities for the scenario with shadowing, which is included in the vast majority of cellular propagation models.\vspace{0.5em}

\noindent \emph{Insights into the expected gains due to collaboration}: Lastly, we present results which shed light on the factors affecting the value of collaboration for improving unique localizability. We observe that in the absence of shadowing, it is the separation between collaborators which dictates the benefit received from collaboration. However, when shadowing is present, the dependence of the benefit on the separation is significantly reduced, although the best case gains are very similar. The value of collaboration is then considered with shadowing \emph{and} frequency reuse, and the results show that the gains due to collaboration are drastically improved. This demonstrates that inter-device collaboration (e.g., using device-to-device (D2D) communication in LTE) is potentially very powerful in cellular networks.
%!TEX root = paper_2.tex

\section{System Model}\label{Sec:SystemModel}
We now formally describe the system model. The key notation presented in this section and used throughout this work is summarized in Table~\ref{Table:Notation}.

\begin{table}
\centering
\caption{Summary of Key Notation}
\label{Table:Notation}
\begin{tabular}{@{}c@{\hspace{2em}}l@{}}
\toprule
\textbf{Notation} & \textbf{Description} \\
\midrule
$\alpha$ & Path loss exponent ($\alpha > 2$) \\
$\Vert \nbz \Vert$ & $\ell_2$-norm of vector $\nbz$ \\
$\origin$ & Origin (location of the typical user) \\
$\PPPa$/$\PPPb$ & PPP of BS/MD locations (independent)\\
$\pppai$/$\pppbi$ & Density of $\PPPa$/$\PPPb$ \\
$\threshold$ & Target SINR \\
$\ncalA \backslash \ncalB$ & The set (or area) $\ncalA$ excluding $\ncalB$ \\
$\vert \ncalA \vert$ & The Lebesgue measure of region $\ncalA$ \\
$\nbbG = (V,E)$ & The graph $\nbbG$ consisting of vertices $V$ and edges $E$ \\
$\ncalS' \subseteq \ncalS$ & $\ncalS'$ is a subset of set $\ncalS$ \\
$\ncalS' \subsetneq \ncalS$ & $\ncalS'$ is a proper (or strict) subset of set $\ncalS$ \\
$\nbbN_\nbz$ & The number of hearable BSs at MD $\nbz$ \\
$\ncalS_{\nbz}^{[\ell]}$ & The set of $\ell$ BSs with highest received SINRs at MD $\nbz$ \\
$\vert \ncalS \vert$ & The cardinality of set $\ncalS$ \\
$\indicator(\cdot)$ & Indicator function, 1/0 when its argument is true/false \\
\bottomrule
\end{tabular}
\end{table}

\subsection{Base station and mobile device locations}

The locations of the BSs and MDs are modeled using two independent homogeneous PPPs $\PPPa,~\PPPb~\in~\R^2$ with densities $\pppai,~\pppbi$~\cite{Stoyan1995, Haenggi2013}, respectively. If the interference is treated as noise at the receiver, the most appropriate metric that captures link quality is the signal-to-interference-plus-noise ratio (SINR). For the link from some BS $x \in \PPPa$ to some MD $z \in \PPPb$, the SINR can be expressed as:
\begin{equation}
\SINR_{x \to z} = \frac{P \ncalF_{x \to z} \|x - z\|^{-\alpha}}{\sum_{\substack{y \in \PPPa\\y \neq x}} P \ncalF_{y \to z} \|y - z\|^{-\alpha} + \N},
\label{Eq:P2:SINR}
\end{equation}
where $P$ is the transmit power, $\ncalF_{g \to h}$ denotes the slow fading coefficient due to shadowing affecting the signal from BS $g$ to MD $h$, $\alpha > 2$ is the pathloss exponent, and $\N$ is the noise variance at the receiver.

In order to improve the hearability of far away BSs, positioning systems typically have to work at lower SINRs than communications systems, thereby necessitating the need for some form of processing gain, which will depend upon system parameters such as the signal integration time. As a side effect, the processing gain techniques employed are assumed to average out the effect of small-scale fading. Thus, the SINR expression in \eqref{Eq:P2:SINR} does not contain a fast fading term, which is consistent with common models for evaluating MANET, WSN, and even cellular positioning performance~\cite{R1-091443}.

\subsection{Base station participation}\label{Sec:BSSelection}

For localization, it is well-known that including an increasing number of BSs in the localization procedure results in a general improvement in positioning accuracy. Thus, we assume that for purposes of localization, a device will take advantage of as many BSs as it can successfully detect (or \emph{hear}), i.e., all BSs whose signals arrive with some minimum link quality. Specifically, a MD $z$ includes a BS $x$ in its localization procedure when
\begin{equation}
\SINR_{x \to z} \geq \threshold,
\end{equation}
where $\threshold$ is the SINR threshold (prior to any processing gain) above which the signals from the BSs must arrive in order for them to successfully contribute to the localization procedure (i.e., this is the \emph{hearability condition}). Note that in the presence of shadowing, the set of included BSs at $z$ will not necessarily correspond to the set of BSs which are geographically closest to $z$.

\subsection{The collaboration-extended localization procedure}

In this paper, we consider the impact of extending a classical location estimation procedure with a single collaborative ranging observation. Formally, we define a classical localization procedure as one where an unlocalized device communicates only with BSs, gathers position-related observations (e.g., RSS, TOA, or TDOA), and solves the resulting single-location estimation problem to determine its location. For MANETs and WSNs, this type of setup corresponds to that employed in \cite{Daneshgaran2007}. For cellular networks, this corresponds to any downlink positioning method, e.g., OTDOA~\cite{Fischer2014}. 
%!TEX root = paper_2.tex

\section{Unique Localizability}\label{P2:Sec:UniqueLoc}

The first objective in any localization system is to make sure that the devices to be located are \emph{uniquely localizable}~\cite{Eren2004, Goldenberg2005}, which is defined next.

\begin{definition}[Unique device localizability]\label{Def:udl}
A mobile device is uniquely localizable if an estimate of the device's location can be found without ambiguity. In the noiseless case, this means that there can only be one solution to the set of non-linear equations that relate the observations to the unknown position. In the noisy case, this means that there is a single global minimum to the appropriate cost function.
\end{definition} 

For classical positioning techniques based on observations between the MD and the BSs, it is widely-accepted that the unique localizability condition simplifies to whether or not a mobile device is able to hear a sufficient number of BSs. Conventional minimum values on the number of BSs required to guarantee the presence of an unambiguous solution to the localization problem in the $\R^2$ plane are 2, 3, and 4 for triangulation (e.g., AOA), trilateration (e.g., TOA and RSS), and multilateration (e.g., TDOA) techniques, respectively. The unique device localizability conditions are much less straightforward for collaborative networks and require a topological analysis of the network as a whole~\cite{Yang2012}. Naturally, researchers have also been interested in the conditions required for an entire network to be uniquely localizable.

\begin{definition}[Unique network localizability]
A network is uniquely localizable when \emph{all} of the devices within the network are uniquely localizable.
\end{definition}

Interestingly and perhaps counterintuitively, the necessary and sufficient conditions for \emph{all} devices in a network to be uniquely localizable using range-based observations were found prior to the corresponding conditions for the \emph{individual} devices. Consider a network consisting of $C$ collaborating MDs and $B \geq 3$ unique noncollinear BSs to which the MDs are connected. Let the grounded network graph\footnote{A grounded network graph differs from a traditional network graph in that additional edges are introduced between all pairs of immovable vertices (i.e., BSs) in order to reflect the rigidity among these vertices.} $\nbbG = (V,E)$ be the graph whose vertices $V$ correspond to the $N = B+C$ network nodes and whose edges $E$ represent all wirelessly-connected pairs (BS~$\rightarrow$~MD and MD~$\leftrightarrow$~MD) as well as all BS pairs. Now, the necessary and sufficient conditions for unique network localizability with perfect ranging observations in $\R^2$ are as follows~\cite{Goldenberg2005}:

\begin{enumerate}
\item[\bf C1]\label{Cond:1} (Rigidity) The grounded graph $\nbbG$ must contain a total of $2N-3$ independent edges. Using Laman's condition \cite{Laman2002}, we can restate this condition as follows: there must exist some graph $\nbbG' = (V, E' \subseteq E)$ where $\vert E' \vert = 2N-3$ for which there are no subgraphs $\nbbG'' = (V'' \subseteq V, E'' \subseteq E')$ where $\vert E'' \vert > 2\vert V'' \vert-3$ edges.
\item[\bf C2]\label{Cond:2} (Triconnectedness) Every vertex $V$ in $\nbbG$ must be the endpoint of at least 3 edges in $E$.
\item[\bf C3]\label{Cond:3} (Reduntant rigidity) If any edge of $\nbbG$ is removed, the ensuing graph must remain rigid. In other words, all subgraphs $\nbbG' = (V, E' \subsetneq E)$ where $\vert E' \vert = \vert E \vert - 1$ must satisfy Condition C1 above.
\end{enumerate}

First, Condition~C1 removes \emph{graph flexibility}, which is defined as the ability to continuously vary the node locations while still satisfying all edge length constraints from the ranging observations. Next, Condition~C2 removes the possibility of pairs of devices being reflected in such a way as to still satisfy all edge length constraints (i.e., \emph{flip ambiguities}). Lastly, Condition~C3 removes the possibility of \emph{flex ambiguities}, i.e., that upon the removal of an edge, the graph loses its rigidity and becomes flexible allowing nodes to be repositioned in such a way which again satisfies all edge length constraints (including the removed edge constraint, which can then be reinserted).

\begin{remark}
As mentioned in~\cite{Goldenberg2005}, the above conditions provide (i) a generic characterization of unique localizability and (ii) assume error-free ranging observations. Regarding (i), the conditions hold for almost all configurations of network devices placed using our PPP models, since the randomization causes degenerate configurations to appear with zero probability in a continuous space~\cite{Goldenberg2005}. Regarding (ii), we note that this is essentially required in order to derive these graph-theoretic localizability conditions. While errors \emph{may} introduce degenerate cases, we will assume that they do not introduce additional global minima into the cost function.
\end{remark}

For this specific study, focusing on how the collaboration of one MD with another MD impacts unique localizability, we can use the above conditions for unique network localizability to obtain the necessary and sufficient conditions for unique device localizability (which we will henceforth also just call \emph{localizability}) using ranging observations. These conditions are formally presented in the following proposition.

\begin{proposition}[Two-device collaborative localizability using ranging observations]\label{Prop:P2:RangingConditions}
Using ranging observations, a device $\da$ capable of collaborating with a second device $\db$ is \emph{uniquely localizable} in $\R^2$ iff one of the following conditions is met:
\begin{enumerate}
\item[\bf L1] $\da$ is directly connected to at least three BSs
\item[\bf L2] $\da$ is directly connected to two BSs, $\db$ is directly connected to \emph{at least} two BSs, and combined, $\da$ and $\db$ are connected to at least three \emph{unique} noncollinear BSs.
\end{enumerate}
\end{proposition}

\begin{figure}
\centering
\subfloat[Three unique BS example.]{\label{Fig:ThreeBeaconScenario}\includegraphics[width=.8\figurewidth]{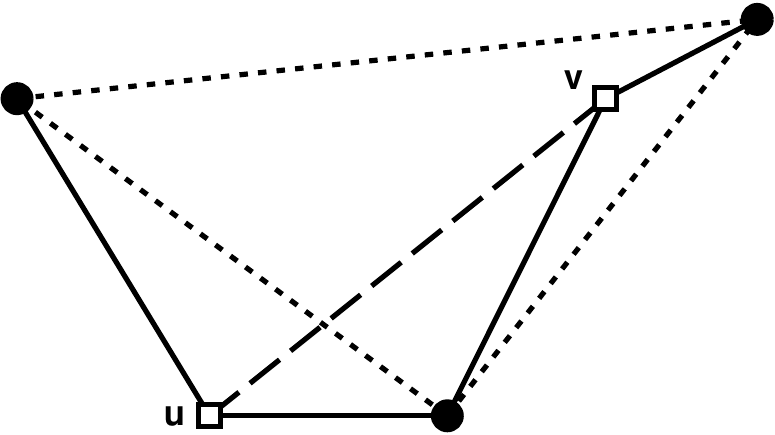}} \\
\subfloat[Four unique BS example.]{\label{Fig:FourBeaconScenario}\includegraphics[width=.7\figurewidth]{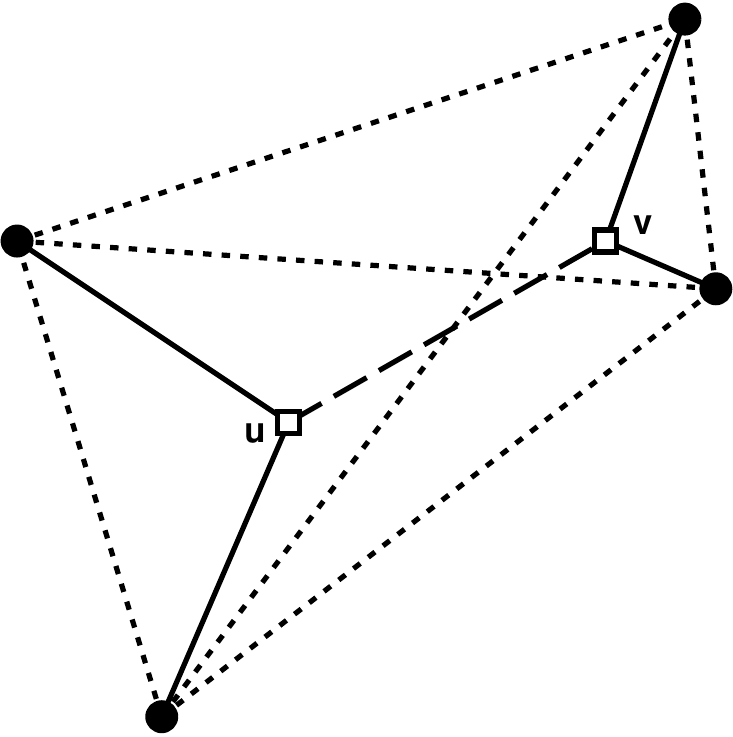}}
\caption{The grounded graphs for the smallest scenarios satisfying Condition~L2, with both MDs (represented by the hollow squares) connected to exactly two BSs (represented by the filled circles) and (a) three and (b) four unique BSs combined. The solid lines represent the BS~$\rightarrow$~MD links, the dashed line is the collaborative MD~$\leftrightarrow$ MD~link, and the dotted lines are not actual wireless links, but instead represent the BS to BS edges included in grounded graphs.}
\label{Fig:L2scenarios}
\end{figure}

%\begin{figure}
%\centering
%\subfloat[Three unique BS example.]{\label{Fig:ThreeBeaconScenario}\includegraphics[width=.8\figurewidth]{ThreeBeaconScenario}} \qquad
%\subfloat[Four unique BS example.]{\label{Fig:FourBeaconScenario}\includegraphics[width=.7\figurewidth]{FourBeaconScenario}}
%\caption{The grounded graphs for the smallest scenarios satisfying Condition~L2, with both MDs (represented by the hollow squares) connected to exactly two BSs (represented by the filled circles). The solid lines represent the BS~$\rightarrow$~MD links, the dashed line is the collaborative MD~$\leftrightarrow$ MD~link, and the dotted lines are not actual wireless links, but instead represent the BS to BS edges included in grounded graphs.}
%\label{Fig:L2scenarios}
%\end{figure}

\begin{proof}
For Condition~L1, consider the smallest scenario which satisfies the condition, i.e., where a lone MD $\da$ is connected to exactly three BSs. It is a trivial exercise to show that this scenario meets Conditions~C1-C3 for network localizability, which is equivalent to unique device localizability (Definition~\ref{Def:udl}) since $\da$ is the only MD in the network. Adding connections from $\da$ to additional BSs cannot violate any of these conditions. For Condition~L2, consider the smallest scenarios which satisfy the condition, i.e., where both MDs are connected to \emph{exactly} two BSs. There are two possible connectivity scenarios, one with three and one with four unique BSs, as illustrated in Figure~\ref{Fig:L2scenarios}. First, note that the graphs are symmetric about the collaborative link, i.e., $\da$ and $\db$ may be arbitrarily switched without a change in the labeled grounded graph. Thus, since the two devices exhibit a symmetry in terms of graph connectivity, it is not possible to deem one localizable without deeming the other so as well, which implies that device localizability is equivalent to network localizability. From here, it is not difficult to verify that both scenarios in Figure~\ref{Fig:L2scenarios} satisfy Conditions~C1-C3 above. Adding additional BSs connected to $\db$ cannot violate these conditions. Since Condition~L2 leads to network localizability, it clearly also leads to device $\da$'s unique localizability. Now, if $\da$ is connected to fewer than two BSs, it will always be subject to a flip ambiguity due to its lack of triconnectedness (Condition~C2). Lastly, if $\da$ is connected to two BSs and $\db$ is connected to fewer than two BSs, then $\da$ does not belong to a redundantly rigid graph component that includes three vertex-disjoint paths to three BSs,\!\footnote{See the RR3P condition in~\cite{Yang2012}.} meaning that it will again not be uniquely localizable. Thus, we have arrived at the necessary and sufficient localizability conditions for device $\da$ in our problem setup when using range-based observations.
\end{proof}

\begin{remark}[Noncollaborative localizability using ranging observations]\label{Remark:NC_TOA}
Note that Condition~L1 above, which does not involve the secondary device $\db$, is the only condition for (and is thus the definition of) \emph{noncollaborative unique localizability} when using ranging observations. This is easily verified by applying Conditions~C1-C3 to any network with $C=1$ MD and $B \geq 3$ noncollinear BSs.
\end{remark}

\hyphenation{OTDOA}
Until now, all of the previous discussion on localizability has been specific to range-based observations. Since cellular positioning is often accomplished using range-difference observations (e.g., OTDOA in LTE), we are also particularly interested in this setup. Since no localizability conditions exist yet for this setup, we will use the following assumption.

\begin{assumption}[Two-device collaborative localizability using range-difference observations]\label{Assumption:P2:RangeDifferenceConditions}
Using range-difference observations from BS signals, a device $\da$ capable of obtaining a range observation from a second device $\db$ is \emph{uniquely localizable} in $\R^2$ iff one of the following conditions is met:
\begin{enumerate}
\item[\bf D1] $\da$ is directly connected to at least four BSs~\cite{Buehrer12}
\item[\bf D2] $\da$ is directly connected to three BSs, $\db$ is directly connected to \emph{at least} three BSs, and combined, $\da$ and $\db$ are connected to at least four \emph{unique} noncollinear BSs.
\end{enumerate}
\end{assumption}

\begin{remark}[Noncollaborative localizability using range-difference observations]
Paralleling Remark~\ref{Remark:NC_TOA}, Condition~D1 is the well-known condition required to guarantee noncollaborative localizability when using range-difference observations~\cite{Buehrer12}.
\end{remark}

% Removed because the paper is really looking at the conditions which GUARANTEE unique localizability (this avoids any confusion and makes the analysis much more airtight); otherwise, Condition D1 also understates noncollaborative localization (since there are cases when it's possible to be uniquely localizable with only 3 BSs)
%\begin{remark}[Pessimism of Condition~D2]
%Although the true conditions for this setup are not known, we note that, if anything, Condition~D2 is sufficient, but may not be necessary. In other words, using the conditions in Assumption~\ref{Assumption:P2:RangeDifferenceConditions} may lead to the recognition of only a subset of all actual uniquely localizable scenarios. Since Condition~D2 only applies to the collaborative scenario, the benefits from collaboration in range-difference-based localization are likely to be somewhat understated throughout this paper.
%\end{remark}

\section{Improvement in Localizability through Collaboration}\label{Sec:P2:LocImprovement}

Using the conditions presented in the previous section, we now move forward with our analysis of the impact a single collaborative link has on a device's ability to locate itself without ambiguity. Let $\nbbN_\nbz = \sum_{x \in \PPPa} \indicator(\SINR_{x \to \nbz} \geq \threshold)$ represent the number of BSs hearable at some device $\nbz$, where we recall that $\threshold$ is the SINR threshold for successfully detecting BS signals. If $\Lnc$ represents the event that device $\da$ is capable of localizing itself using only its BS connections, then
\begin{align}
\P(\Lnc) = \P(\nbbN_\da \geq \ell+1),\label{Eq:P2:PLnc}
\end{align}
where $\ell=2$ in the case of ranging observations to BSs (Condition~L1) and $\ell=3$ in the case of range-difference observations to BSs (Condition~D1). Now, let $\ncalS_{\nbz}^{[\ell]}$ represent the set of $\ell$ BSs whose signals arrive with the highest SINRs at some device $\nbz$. Then, the probability of $\Lc$, the event that device $\da$ is localizable when capable of collaborating with a second device $\db$, is
\begin{alignat}{2}
\P\left(\Lc\right)
&= \P\left(\Lnc\right) &&+ \sum_{n=\ell}^\infty \P\left(\nbbN_\da = \ell, \nbbN_\db = n, \left\Vert \seta{\ell} \cup \setb{n} \right\Vert \geq \ell+1\right) \notag\\
%&\stackrel{(a)}{=} \P\left(\Lnc\right) &&+ \sum_{n=\ell}^\infty \P\left(\nbbN_\da = \ell, \nbbN_\db = n, \left\Vert \seta{\ell} \cup \setb{n} \right\Vert \geq \ell+1, \seta{\ell} \neq \setb{\ell}\right)\notag\\
%&&&+ \sum_{n=\ell}^\infty \P\left(\nbbN_\da = \ell, \nbbN_\db = n, \left\Vert \seta{\ell} \cup \setb{n} \right\Vert \geq \ell+1, \seta{\ell} = \setb{\ell}\right)\notag \\
&\stackrel{(a)}{=} \P\left(\Lnc\right) &&+ \sum_{n=\ell}^\infty \P\left(\nbbN_\da = \ell, \nbbN_\db = n, \left\Vert \seta{\ell} \cup \setb{n} \right\Vert \geq \ell+1\middle\vert \seta{\ell} \neq \setb{\ell}\right)\P\left(\seta{\ell} \neq \setb{\ell}\right)\notag\\
&&&+ \sum_{n=\ell}^\infty \P\left(\nbbN_\da = \ell, \nbbN_\db = n, \left\Vert \seta{\ell} \cup \setb{n} \right\Vert \geq \ell+1\middle\vert \seta{\ell} = \setb{\ell}\right)\P\left(\seta{\ell} = \setb{\ell}\right) \notag \\
&\stackrel{(b)}{=} \P\left(\Lnc\right) &&+ \sum_{n=\ell}^\infty \P\left(\nbbN_\da = \ell, \nbbN_\db = n \middle\vert \seta{\ell} \neq \setb{\ell}\right)\P\left(\seta{\ell} \neq \setb{\ell}\right)\notag\\
&&&+ \sum_{n=\ell+1}^\infty \P\left(\nbbN_\da = \ell, \nbbN_\db = n \middle\vert \seta{\ell} = \setb{\ell}\right)\P\left(\seta{\ell} = \setb{\ell}\right)  \notag \\
&\stackrel{(c)}{=} \P\left(\Lnc\right) &&+ \sum_{n=\ell}^\infty \P\left(\nbbN_\da = \ell, \nbbN_\db = n \middle\vert \seta{\ell} \neq \setb{\ell}\right)\P\left(\seta{\ell} \neq \setb{\ell}\right)\notag\\
&&&+ \sum_{k=1}^\infty \P\left(\nbbN_\da = \ell, \nbbN_\db = \ell+k \middle\vert \seta{\ell} = \setb{\ell}\right)\P\left(\seta{\ell} = \setb{\ell}\right),\label{Eq:PLc1}
\end{alignat}
where $\ell = 2, 3$ for range and range-difference observations, respectively, $(a)$ follows from the law of total probability and Bayes' rule, $(b)$ follows from (i) the fact that $\nbbN_\da = \ell$, $\nbbN_\db = n \geq \ell$, and $\seta{\ell} \neq \setb{\ell}$ imply that $\left\Vert \seta{\ell} \cup \setb{n} \right\Vert \geq \ell+1$ and (ii) the fact that $\seta{\ell} = \setb{\ell}$ implies that $\left\Vert \seta{\ell} \cup \setb{\ell} \right\Vert \not\geq \ell+1$, and $(c)$ follows from a simple rewriting of the lower limit in the second summation.

\subsection{The no shadowing case}

First, we consider the scenario without shadowing. When shadowing is absent, i.e., $\ncalF_{g \to h} = 1$ in \eqref{Eq:P2:SINR} for all $g$ and $h$, ranking the BSs by decreasing SINRs is equivalent to ranking them by increasing distances from the receiver. In other words, there is a strong correlation between the hearable BSs at two nearby devices. Following this train of thought, we now note that in \eqref{Eq:PLc1},
\[
\P\left(\nbbN_\da = \ell, \nbbN_\db = \ell+k\middle\vert \seta{\ell} = \setb{\ell}\right) \to 0
\]
quickly as $k$ increases. Intuitively, a hearability mismatch of $k$ BSs is unlikely, even for small values of $k$,  when the closest $\ell$ BSs to $\da$ and $\db$ are conditioned to be the same. We now remove this term by letting $\P\left(\nbbN_\da = \ell, \nbbN_\db = \ell+k\middle\vert \seta{\ell} = \setb{\ell}\right)=0$ since $k \geq 1$ and approximate \eqref{Eq:PLc1} as
\begin{equation}
\P\left(\Lc\right) \approx \P\left(\Lnc\right) + \sum_{n=\ell}^\infty \P\left(\nbbN_\da = \ell, \nbbN_\db = n \middle\vert \seta{\ell} \neq \setb{\ell}\right)\P\left(\seta{\ell} \neq \setb{\ell}\right). \label{Eq:PLc2}
\end{equation}
Lastly, we make the following assumption which will simplify the analysis of the conditional joint hearability probability.

\begin{assumption}[Independent base station hearability]\label{Assumption:indep_hearability}
When two devices have different sets of $\ell$ strongest base stations, their joint hearability probability may be calculated as the product of their individual hearability probabilities. Mathematically, this means
\[
\P\left(\nbbN_\da = m, \nbbN_\db = n \middle\vert \seta{\ell} \neq \setb{\ell}\right) = \P\left(\nbbN_\da = m\right)\P\left(\nbbN_\db = n\right),
\]
where $\ell = 2, 3$ is used for range-based and range-difference-based localization, respectively.
\end{assumption}

Under Assumption~\ref{Assumption:indep_hearability}, we arrive at the following final expression for $\P(\Lc)$ in \eqref{Eq:PLc2}:
\begin{align}
\P\left(\Lc\right)
&\approx \P\left(\Lnc\right) + \sum_{n=\ell}^\infty \P\left(\nbbN_\da = \ell\right)\P\left(\nbbN_\db = n\right)\P\left(\seta{\ell} \neq \setb{\ell}\right) \notag \\
&= \P\left(\nbbN_\da \geq \ell+1\right) + \P\left(\nbbN_\da = \ell\right)\P\left(\nbbN_\db \geq \ell\right)\P\left(\seta{\ell} \neq \setb{\ell}\right). \label{Eq:PLc3}
\end{align}
Exact expressions for the hearability terms in~\eqref{Eq:PLc3} (i.e., all terms except $\P(\seta{\ell} \neq \setb{\ell})$) are presented in~\cite{Keeler2013}. The exact calculations, however, are extremely involved and time consuming, leading us to employ the approximations presented in~\cite{Schloemann2015c}, which are nearly indistinguishable from truth and can be calculated instantly. Specifically, $\P(\nbbN_\nbz \geq L) = \pl(1,1,\alpha,\threshold,1,\pppai)$, where the right-hand term is presented in Theorem~2 of~\cite{Schloemann2015c}. In the following section, we will derive an exact expression for the remaining term.

\subsection{Probability that two devices share the same set of closest BSs}\label{Sec:2dSameBSs}

Recall that MD $\da$ is the device whose localizability is being directly considered and MD $\db$ is a secondary device with whom $\da$ is able to collaborate. Due to the stationarity of homogeneous PPPs and Slivnyak's Theorem~\cite{Haenggi2013}, the statistics of $\PPPa$ and $\PPPb$ are unaffected by the arbitrary placement of a finite number of MDs in $\R^2$. Thus, without loss of generality, let $\da$ be located at the origin $\origin$ and $\db$ be a random distance $D$ away and located at $\db = [D\ 0]^\T\!\!\!,$ where ${}^\T$ is the matrix transpose operator. We begin our derivation of $\P(\seta{\ell} \neq \setb{\ell})$ without considering the distribution of $D$, but rather by conditioning on $D = d$. In order to proceed, we first need to understand the following shape.

\begin{figure}
\centering
\includegraphics[width=0.8\figurewidth]{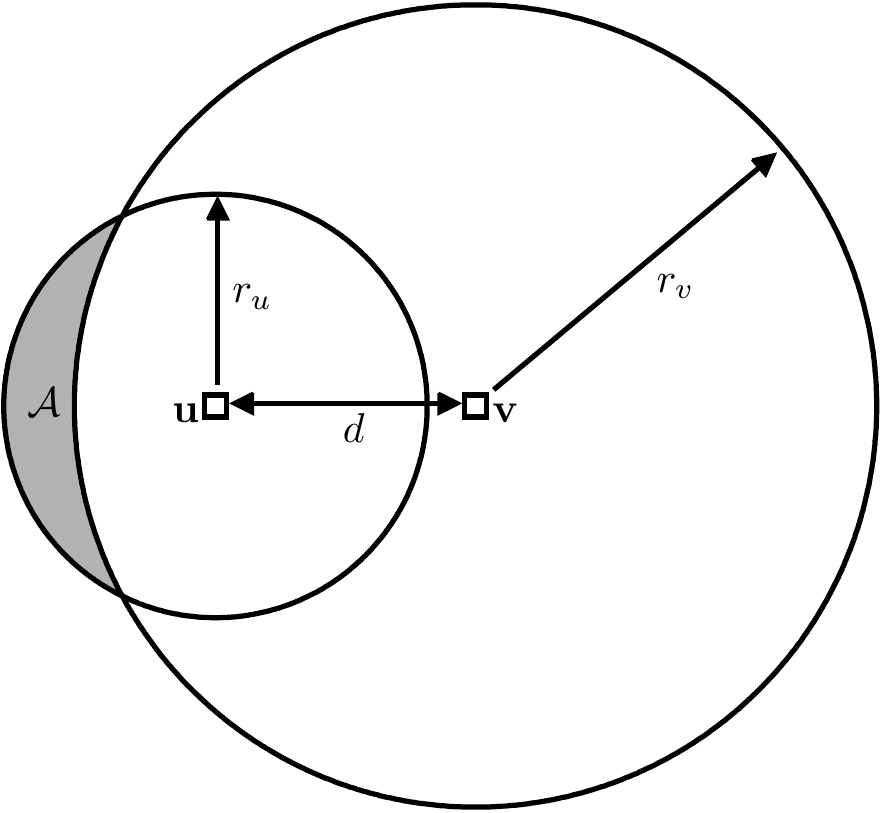}
\caption{\textsc{A lune} is formed by the region of one circle which is outside its intersection with another partially-overlapping circle. The area of the lune, $\ncalA$, depends on $r_u$, $r_v$, and $d$ as described in \eqref{Eq:LuneArea}.}
\label{Fig:lune}
\end{figure}
\begin{definition}[Lune]
Consider two partially-overlapping circles with radii $r_u$ and $r_v$ whose centers are separated by distance $d$, as shown in Figure~\ref{Fig:lune}. Region $\ncalA$ is called a \emph{lune} and its area is~\cite{LuneArea}
\begin{align}
\Alune(r_u, r_v, d) = &\frac{1}{2}\sqrt{(r_u+r_v+d)(r_v+d-r_u)(d+r_u-r_v)(r_u+r_v-d)} \notag \\
&+ r_u^2 \sec^{-1}\left(\frac{2 d r_u}{r_v^2 - r_u^2 - d^2}\right) - r_v^2 \sec^{-1}\left(\frac{2 d r_v}{r_v^2 + d^2 - r_u^2}\right).\label{Eq:LuneArea}
\end{align}
\end{definition}

Next, we present a set of lemmas which are necessary to characterize the probability that two devices $\da$ and $\db$, separated by distance $d$, share the same set of $\ell$ closest BSs. Two lemmas are necessary in order to capture the two key geometric conditions which arise: when the $\ell\th$ farthest BSs to $\da$ and $\db$ are (i) the same and (ii) different.

\begin{lemma}\label{Lemma:sameLth}
The probability that $\da$ and $\db$, separated by distance $d$, have the same set of $\ell$ closest base stations, while also having the same $\ell\th$ closest BS, is

\noindent $\P(\seta{\ell} = \setb{\ell}, \seta{\ell-1} = \setb{\ell-1} | D=d) = $
\begin{align}
&\frac{2}{\pi\Gammal} \int_0^\infty \int_0^\pi \left( \frac{\pi r^2-\Alune(r, \bL, d)}{\pi r^2} \right)^{\ell-1} \notag \\
&\times e^{-\pppai(\Alune(\bL, r, d)+\pi r^2)} \frac{(\pppai \pi r^2)^\ell}{r}\ \d\theta\ \d{r}.
\end{align}
\end{lemma}
\begin{proof}
See Appendix~\ref{Proof:sameLth}.
\end{proof}

\begin{lemma}\label{Lemma:diffLth}
The probability that $\da$ and $\db$, separated by distance $d$, have the same set of $\ell$ closest BSs, while differing in their $\ell\th$ closest BSs, is

\noindent $\P(\seta{\ell} = \setb{\ell}, \seta{\ell-1} \neq \setb{\ell-1} | D=d) =$
\begin{align}
&\frac{2(\ell-1)}{\pi\Gammal} \int_0^\infty \frac{1}{\pi r^2} \int_0^{\pi} \int_{\bL}^{d+r} \left( \frac{\pi r^2-\Alune(r,x,d)}{\pi r^2} \right)^{\ell-2} \notag \\
&\times e^{-\pppai (\Alune(x,r,d)+\pi r^2)} \phirange(d,r,x)\frac{x(\pppai \pi r^2)^\ell}{r}\  \d x\ \d \theta\ \d r,
\end{align}
where
\[
\phirange(d,r,x) = 2 \cos^{-1}\!\left(\frac{d^2 + x^2 - r^2}{2 \cdot d \cdot x}\right).
\]
\end{lemma}
\begin{proof}
See Appendix~\ref{Proof:diffLth}.
\end{proof}

\noindent Combining the two lemmas, we arrive at the following theorem.

\begin{theorem}
The probability that two devices $\da$ and $\db$, separated by distance $d$, share the same set of $\ell$ closest BSs is
\begin{align}
\P(\seta{\ell} = \setb{\ell} | D=d) =\ &\P(\seta{\ell} = \setb{\ell}, \seta{\ell-1} = \setb{\ell-1} | D=d)\notag\\
&+ \P(\seta{\ell} = \setb{\ell}, \seta{\ell-1} \neq \setb{\ell-1} | D=d). \label{Eq:P2:Thm1}
\end{align}

\end{theorem}
\begin{proof}
By the law of total probability, the probability that the two mobile devices have the same set of $\ell$ closest BSs is simply the sum of the probabilities presented in Lemmas~\ref{Lemma:sameLth} and~\ref{Lemma:diffLth}.
\end{proof}

\begin{corollary}\label{Corr:P2:1.1}
When devices $\da$ and $\db$ (separated by distance $d$) both successfully hear exactly $\ell$ base stations, the probability that collaboration between them will result in a combined hearability of at least $\ell+1$ unique base stations is
\begin{equation}
\P(\| \seta{\ell} \cup \setb{\ell}\| \geq \ell+1 | D=d) = \P(\seta{\ell} \neq \setb{\ell} | D=d) = 1 - \P(\seta{\ell} = \setb{\ell} | D=d).\label{Eq:P2:Corr1.1}
\end{equation}
\end{corollary}

In order to endow $D$ with a distribution, let us now consider device $\db$ to be the $K\th$ closest MD to $\da$. Since the MDs are modeled according to a homogeneous PPP with density $\pppbi$, it follows from Slivnyak's theorem~\cite{Haenggi2013} that the distribution of the distance from any device to its $K\th$ neighbor, $D = D_K$, is~\cite{Haenggi2005}
\begin{align}
	f_{D_K}(d; K, \pppbi)
	&= e^{-\pppbi \pi d^2 } \frac{2 (\pppbi \pi d^2)^K}{d \GammaK}. \label{Eq:P2:D_K}
\end{align}
Clearly, $K=1$ represents a case of particular interest, i.e., $\da$ collaborates with its closest neighbor. Now, for the general $K\th$ neighbor setup, we arrive at the following theorem.

\begin{theorem}
The probability that device $\da$ and its $K\th$ closest neighboring device $\db$ share the same set of $\ell$ closest base stations is
\begin{align}
\P(\seta{\ell} = \setb{\ell}) = \frac{2}{\GammaK}\int_0^\infty \P(\seta{\ell} = \setb{\ell} | D=\dK) e^{-\pppbi \pi \dK^2 } \frac{(\pppbi \pi \dK^2)^K}{\dK}\ \d\dK. \label{Eq:Thm:KthNeighbor}
\end{align}
\end{theorem}
\begin{proof}
The result is obtained by deconditioning \eqref{Eq:P2:Thm1} on $D = D_K$, i.e.,
\[
\P(\seta{\ell} = \setb{\ell}) = \E_{D_K}\!\left[ \P(\seta{\ell} = \setb{\ell} | D_K) \right].
\]
\end{proof}

\begin{corollary}\label{Corr:P2:2.1}
Conditioned on device $\da$ and its $K\th$ closest neighbor $\db$ both successfully hearing exactly $\ell$ base stations, the probability that collaboration among them will lead to a combined hearability of at least $\ell+1$ unique base stations is
\begin{equation}
\P\left(\| \seta{\ell} \cup \setb{\ell}\| \geq \ell+1\right) = \P\left(\seta{\ell} \neq \setb{\ell}\right) = 1 - \P\left(\seta{\ell} = \setb{\ell}\right).\label{Eq:P2:Corr2.1}
\end{equation}
\end{corollary}

\noindent Finally, \eqref{Eq:P2:Corr2.1} in Corollary~\ref{Corr:P2:2.1} is the exact expression for $\P(\seta{\ell} \neq \setb{\ell})$ in \eqref{Eq:PLc3}, which, when combined with the hearability results in~\cite{Schloemann2015c}, yields $\P(\Lc)$, the probability of unique localizability in the collaborative scenario.

\subsection{The shadowing case}\label{Sec:P2:ShadowingCase}

Now, we consider the unique localizability problem in the presence of log-normal shadowing. The difficulty in analyzing this scenario lies in the fact that, unlike in the no shadowing case, the set of $\ell$ strongest BSs at some device $\nbz$, $\ncalS_{\nbz}^{[\ell]},$ is no longer directly tied to the $\ell$ geographically closest BSs to $\nbz$. Thus, we cannot use the geometric analysis of the previous section for comparing the sets of strongest BSs at two MDs. Instead, we note that $\P\left(\| \seta{\ell} \cup \setb{\ell}\| \geq \ell+1 \middle\vert D=d\right) = \P\left(\seta{\ell} \neq \setb{\ell} \middle\vert D=d\right) \to 1$ for all $d$ as the shadowing standard deviation $\sigma_s$ increases. This behavior is shown in Figure~\ref{Fig:P2_corollary11_lowerbound} for $\ell=2$, $\alpha=4$, and a shadowing correlation of 0.5 between the received signals at $\da$ and $\db$ from the same BS. In order to get an initial tractable expression for the shadowing case, we then use the simplifying assumption that $\P\left(\seta{\ell} \neq \setb{\ell}\right) = 1$ and invoke Assumption~\ref{Assumption:indep_hearability} to arrive at the following approximation of \eqref{Eq:PLc1} for the shadowing case:
\begin{align}
\P\left(\Lc\right) &\approx \P\left(\Lnc\right) + \sum_{n=\ell}^\infty \P\left(\nbbN_\da = \ell\right)\P\left(\nbbN_\db = n\right)\notag\\
&= \P\left(\nbbN_\da \geq \ell+1\right) + \P\left(\nbbN_\da = \ell\right)\P\left(\nbbN_\db \geq \ell\right).\label{Eq:P2:PLcShad}
\end{align}
An exact expression for \eqref{Eq:PLc1} in the shadowing case is significantly more challenging to derive and is outside the scope of this paper. Nevertheless, it will be evident in the following section that \eqref{Eq:P2:PLcShad} is surprisingly accurate.

\begin{figure}
\centering
\includegraphics[width=\figurewidth]{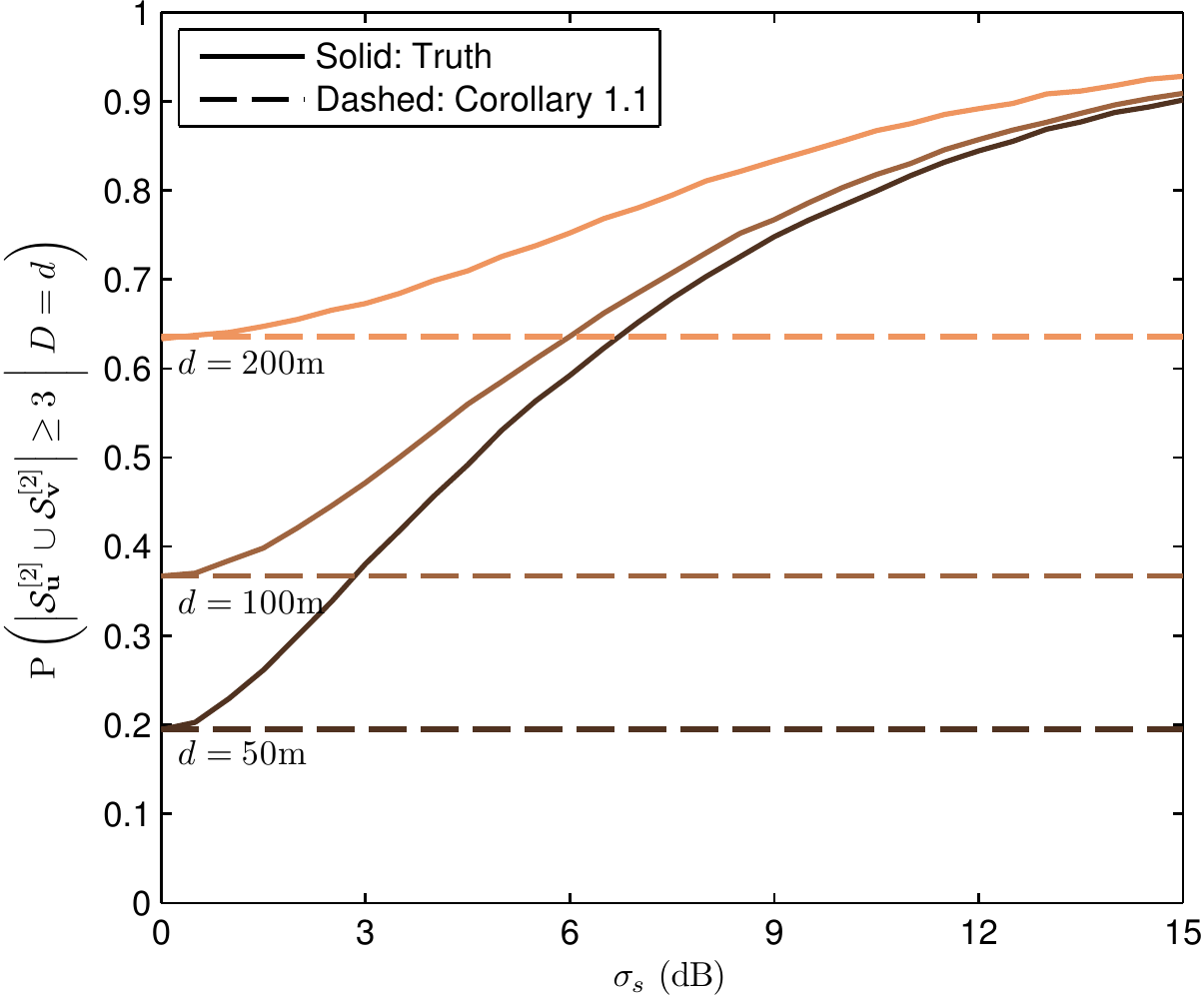}
\caption{\textsc{The Impact of Shadowing} on the probability that collaboration will increase the number of unique BSs involved in the positioning procedure when two collaborators, separated by distance $d$, each hear exactly $\ell=2$ BSs. ($\alpha=4$.)}
\label{Fig:P2_corollary11_lowerbound}
\end{figure}
%!TEX root = paper_2.tex

\section{Numerical Results and Discussion}\label{Sec:NumericalResultsAndDiscussion}
In this section, we present numerical results and use them to draw insights into the value of collaboration for improving unique localizability. We begin by focusing on the no shadowing case and taking a look at the number of unique BSs among the closest BSs at two devices.

\subsection{Sufficient unique base stations versus collaborator separation}

\begin{figure}
\centering
\includegraphics[width=\figurewidth]{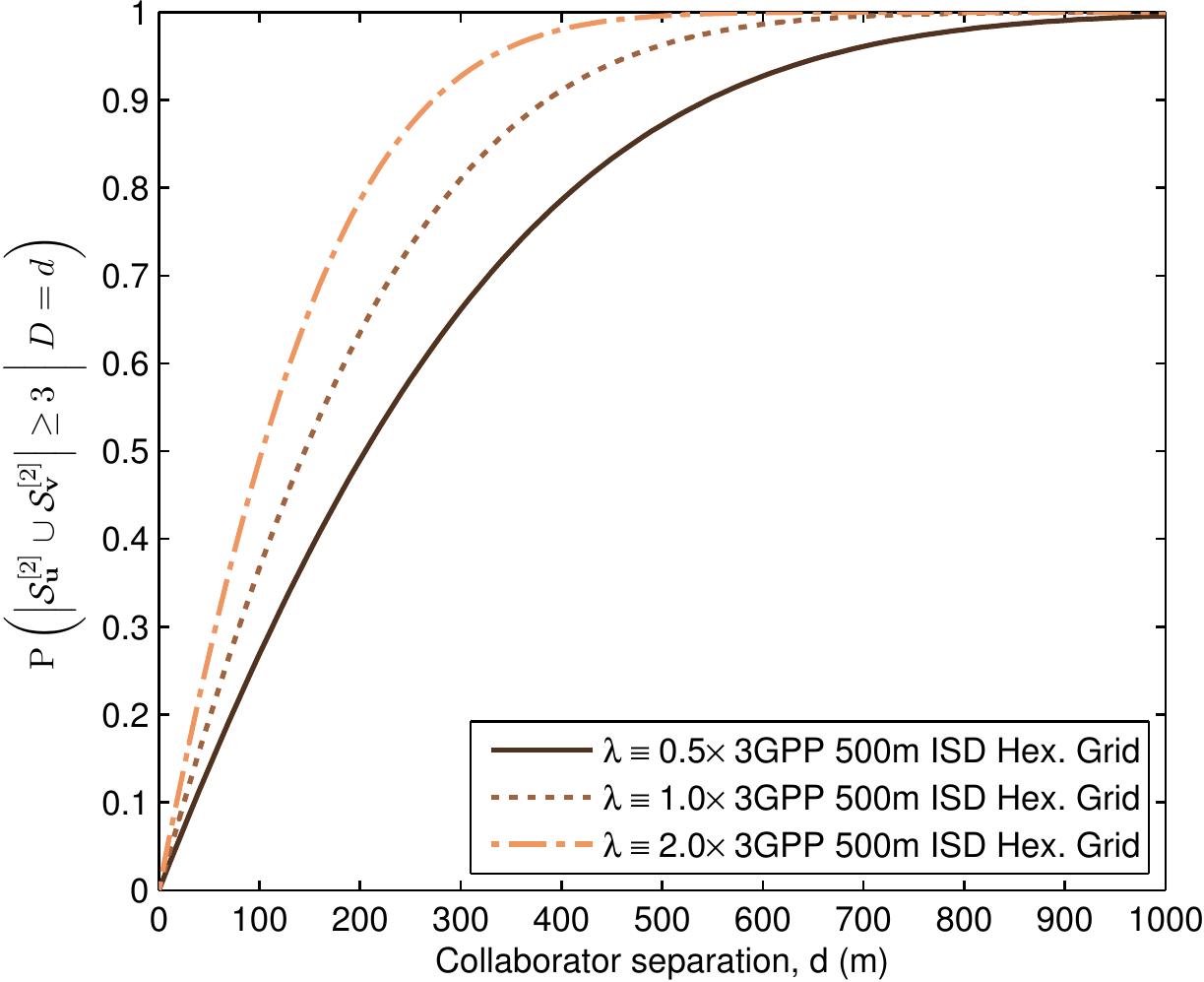}
\caption{\textsc{Uniqueness Among Closest Nodes}: The probability that two devices, separated by distance $d$ and each hearing exactly $\ell=2$ BSs, will benefit in terms of their combined number of unique BSs (Corollary~\ref{Corr:P2:1.1}).}
\label{Fig:P2_corollary11_v_d}
\end{figure}

First, let $\da$ and $\db$ be two devices separated by distance $d$ as described in Section~\ref{Sec:2dSameBSs}. Furthermore, recall that $\ell$ in \eqref{Eq:PLc1} equals 2 and 3 for range-based and range-difference-based localization, respectively. When both devices successfully hear exactly $\ell$ BSs, neither is localizable per the conditions presented in Proposition~\ref{Prop:P2:RangingConditions} and Assumption~\ref{Assumption:P2:RangeDifferenceConditions}. The key differentiator in determining whether collaboration between these devices will be beneficial to localizability is whether or not the two devices hear a combined $\ell+1$ or more unique BSs. When $d$ is fixed, it is the density of the BSs which will affect the probability of obtaining a sufficient number of unique BSs. This is illustrated in Figure~\ref{Fig:P2_corollary11_v_d} for $\ell=2$, where \eqref{Eq:P2:Corr1.1} in Corollary~\ref{Corr:P2:1.1} is plotted versus $d$ for various BS densities $\pppai$. Note that the densities are multiples of the PPP density which results in the same average number of BSs per unit area as an infinite hexagonal grid with 500m intersite distances (ISD). For a fixed separation $d$, it is obvious that a higher BS density leads to a greater likelihood that collaboration will be beneficial in this scenario. While collaboration with farther devices also increases this likelihood, Figure~\ref{Fig:P2_corollary11_v_d} reveals that there is a certain distance beyond which it is not necessary to collaborate.

\subsection{Sufficient unique base stations versus collaborator selection}

\ifx\sidebysidefigures\undefined

\begin{figure}
\centering
\subfloat[Range observations from BSs, $\ell = 2$.]{\label{Fig:P2_corollary21_v_K}\includegraphics[width=\figurewidth]{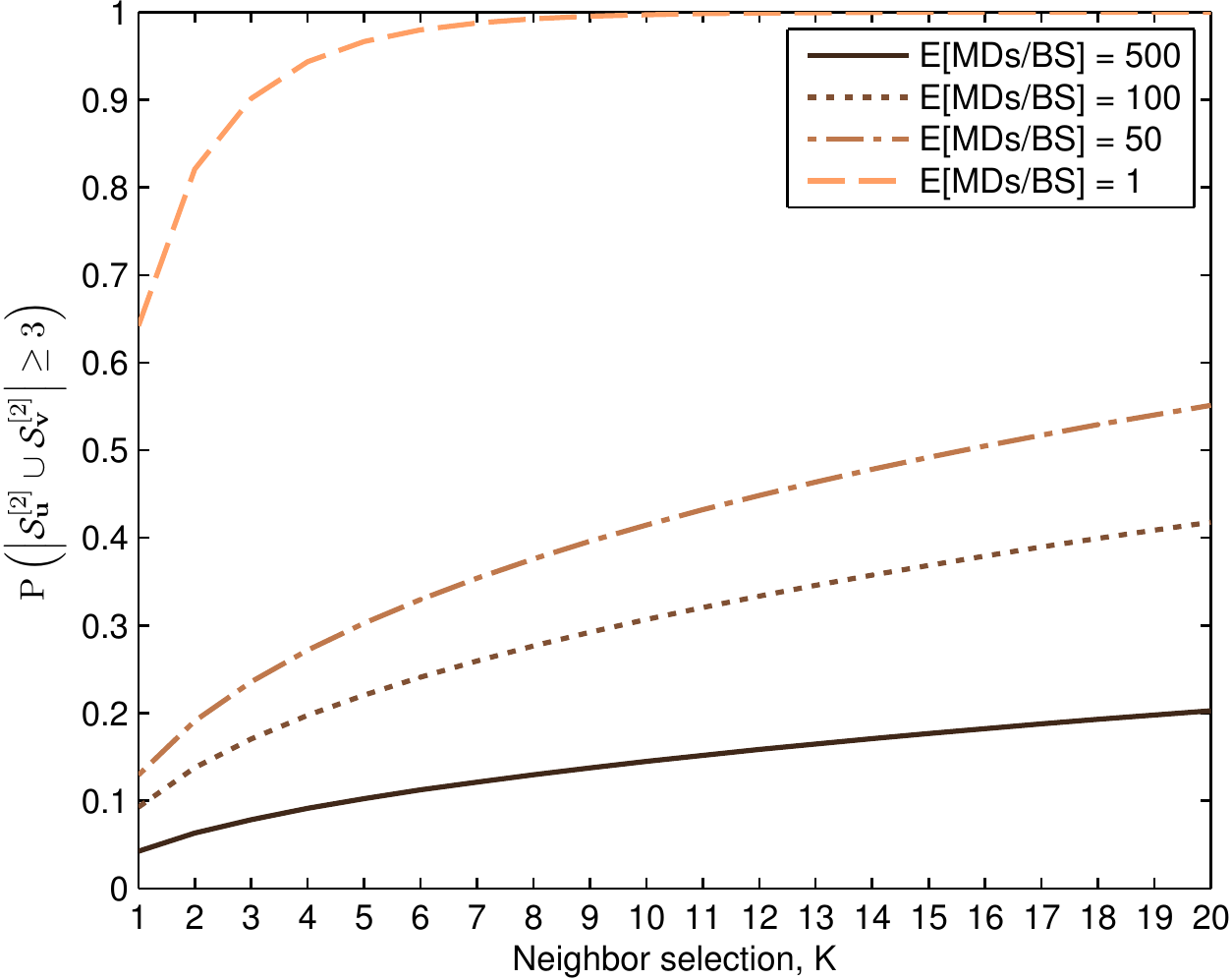}} \\
\subfloat[Range difference observations from BSs, $\ell = 3$.]{\label{Fig:P2_corollary21_v_K_cell}\includegraphics[width=\figurewidth]{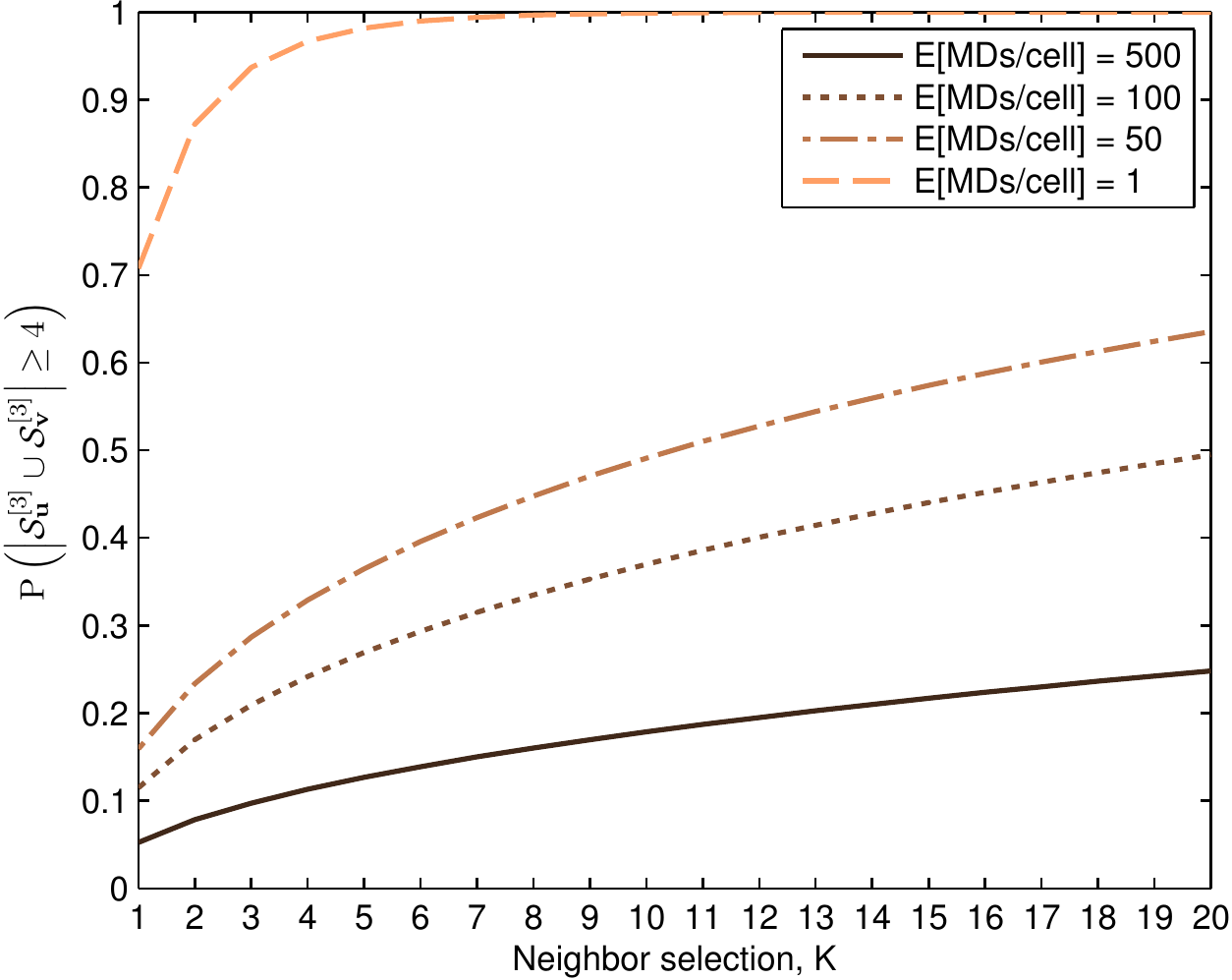}}
\caption{\textsc{The Impact of Neighbor Selection}: The analysis of Corollary~\ref{Corr:P2:2.1} reveals the benefit of selecting farther neighboring devices for increasing their combined number of unique BSs.}
\label{Fig:P2_corollary21}
\end{figure}

\else

\begin{figure}
\centering
\subfloat[Range observations from BSs, $\ell = 2$.]{\label{Fig:P2_corollary21_v_K}\includegraphics[width=\figurewidth]{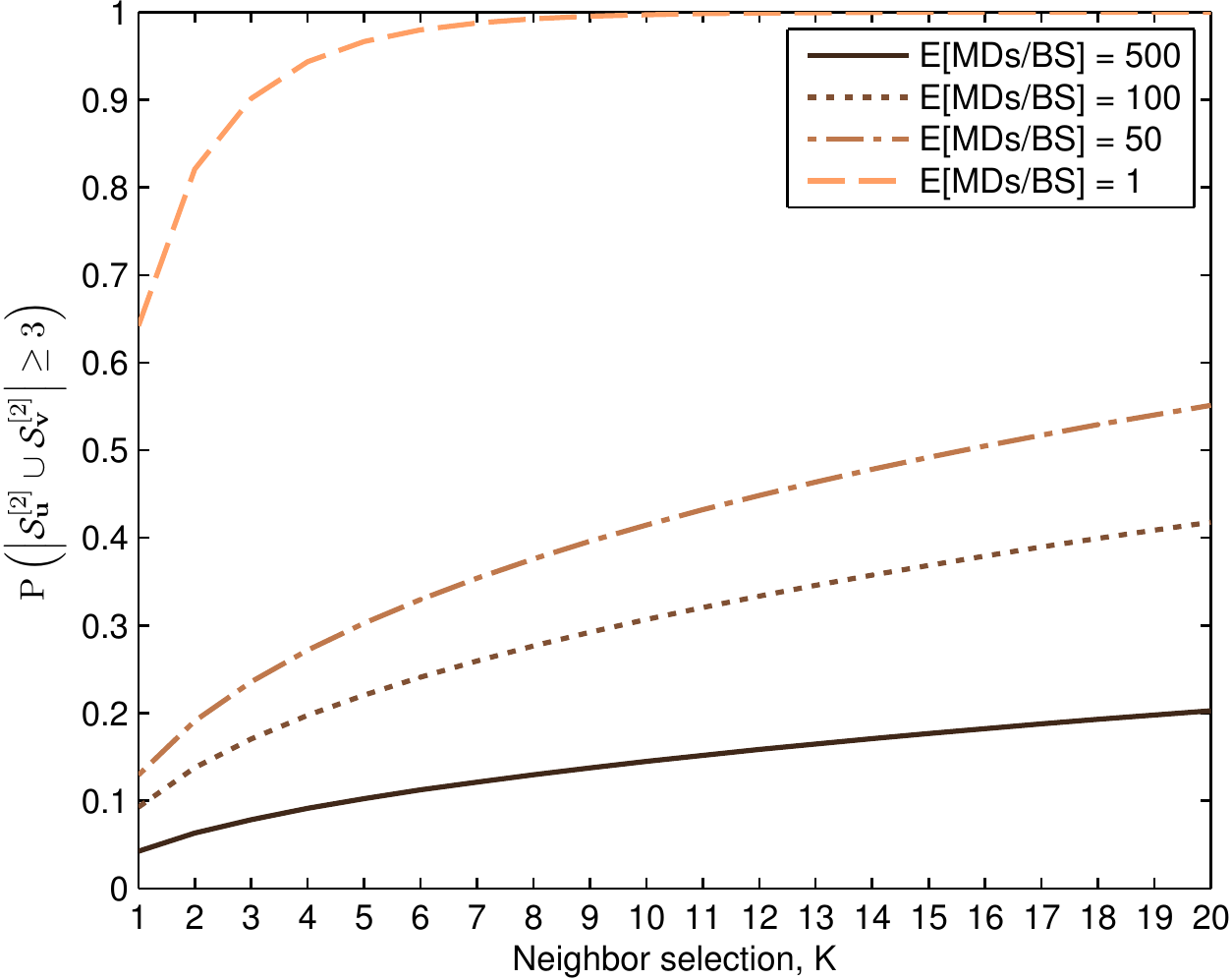}} \qquad
\subfloat[Range difference observations from BSs, $\ell = 3$.]{\label{Fig:P2_corollary21_v_K_cell}\includegraphics[width=\figurewidth]{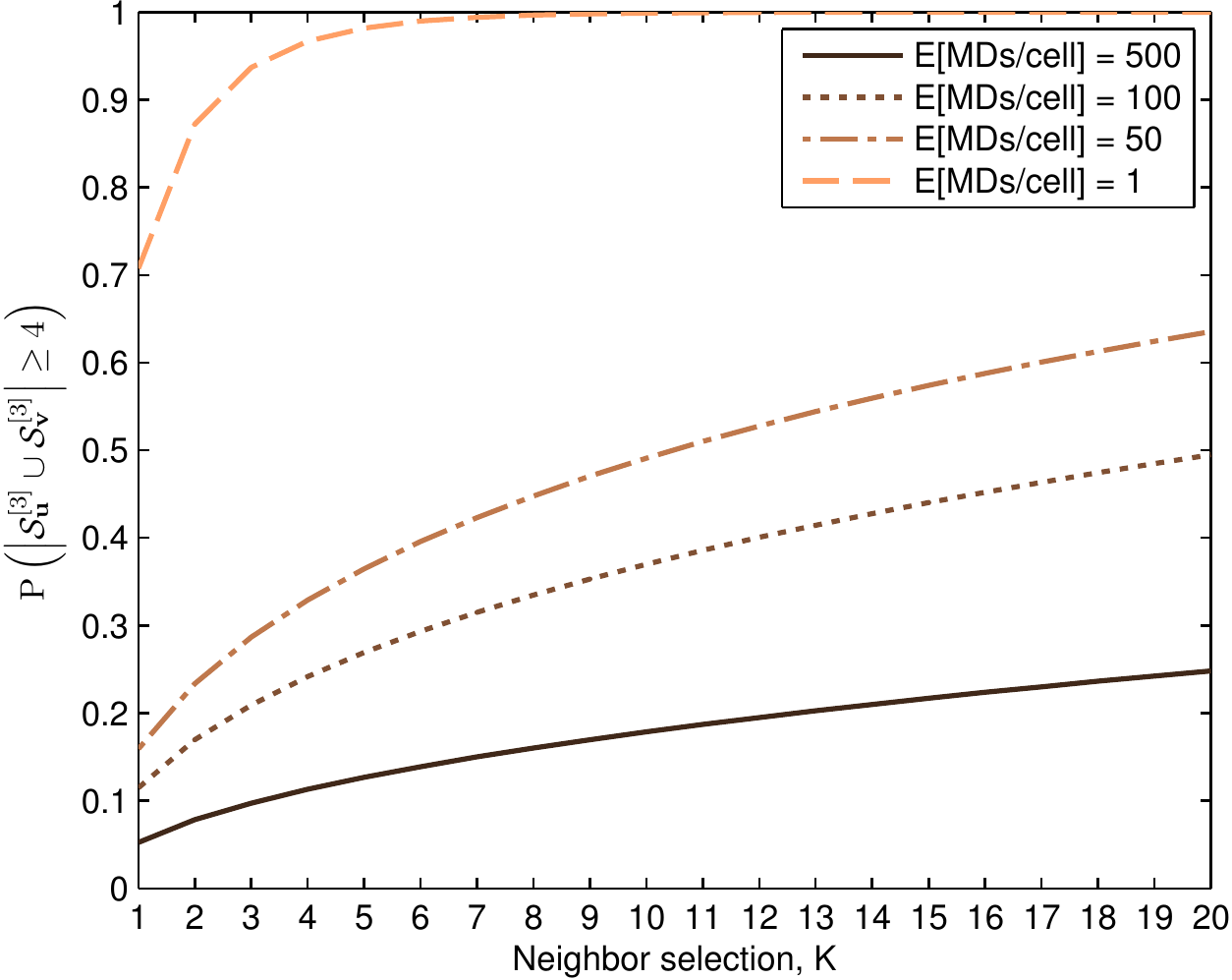}}
\caption{\textsc{The Impact of Neighbor Selection}: The analysis of Corollary~\ref{Corr:P2:2.1} reveals the benefit of selecting farther neighboring devices for increasing their combined number of unique BSs.}
\label{Fig:P2_corollary21}
\end{figure}

\fi

Next, let $\db$ be the $K\th$ neighbor of $\da$ in the PPP of MDs $\PPPb$. Given that both devices successfully hear exactly $\ell$ BSs, Figure~\ref{Fig:P2_corollary21} presents the probabilities that a collaborative link between $\da$ and $\db$ will be beneficial to their unique localizabilities for (a) range and (b) range-difference observations from the BSs for various MD densities (expressed as the average number of MDs per BS or cell). The results show that selecting the closest neighbor, with whom it may likely be easiest to collaborate, is not necessarily a good idea, especially when a MD is in the neighborhood of a large number of other MDs. The nearest neighbors are valuable when MD densities are low, but even then, selecting a farther neighbor is typically more beneficial (up to a point). By revisiting~\eqref{Eq:P2:D_K}, it becomes clear that higher values of $K$ and lower values of $\pppbi$ lead to greater probabilities of \emph{longer} distances separating $\da$ and $\db$. Thus, we see that ultimately, it is the separation between collaborators that is the driving force behind determining the value of collaboration, which leads us to present all subsequent results in light of the distance $d$ separating the collaborating devices.

\subsection{Unique localizability with range and range-difference observations}

\ifx\sidebysidefigures\undefined

\begin{figure}
\centering
\subfloat[Probability of unique localizability]{\label{Fig:P2_localizability_withoutshadowing_vs_d}\includegraphics[width=\figurewidth]{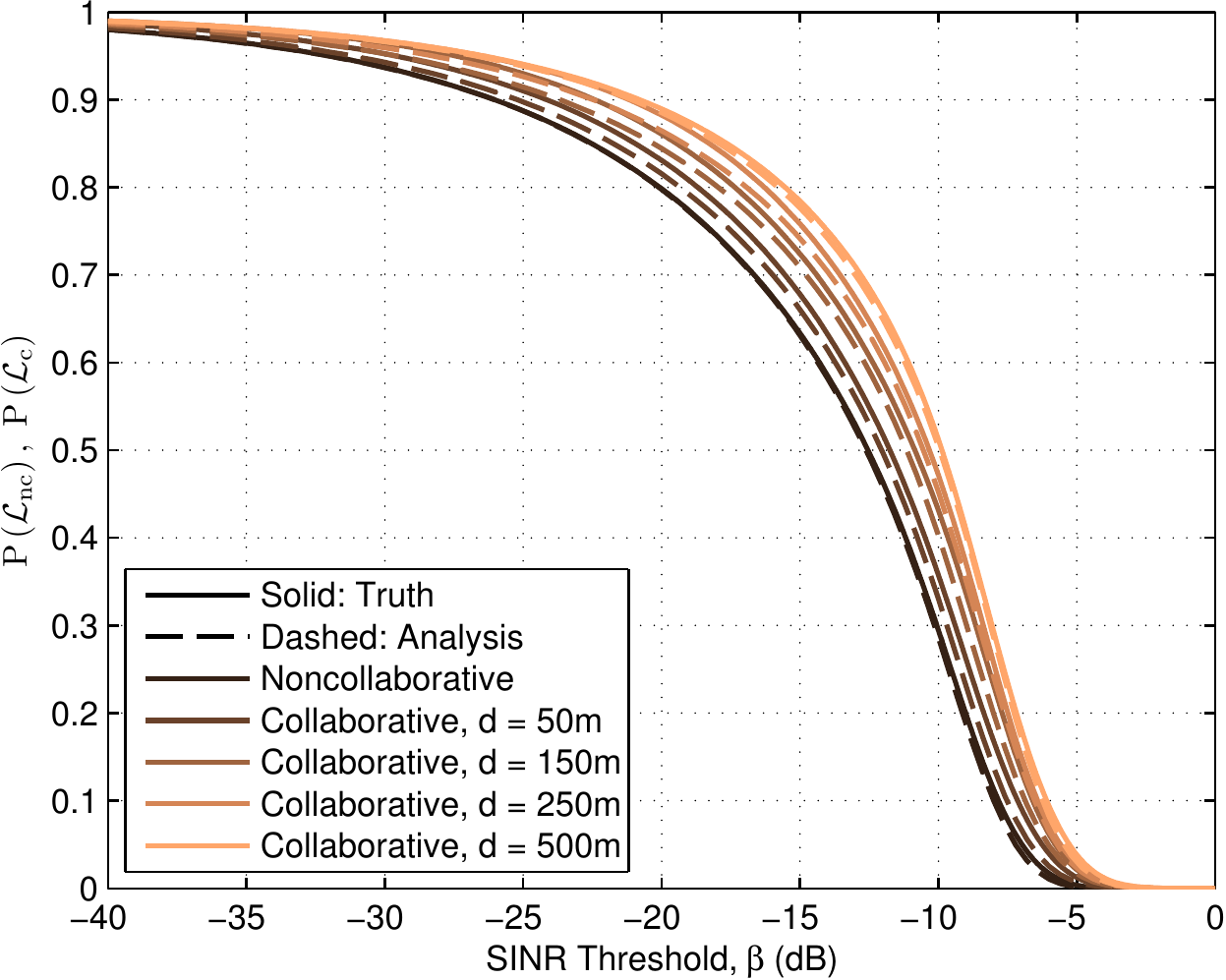}} \\
\subfloat[Absolute increase in unique localizability probability]{\label{Fig:P2_localizability_improvement_withoutshadowing_vs_d}\includegraphics[width=\figurewidth]{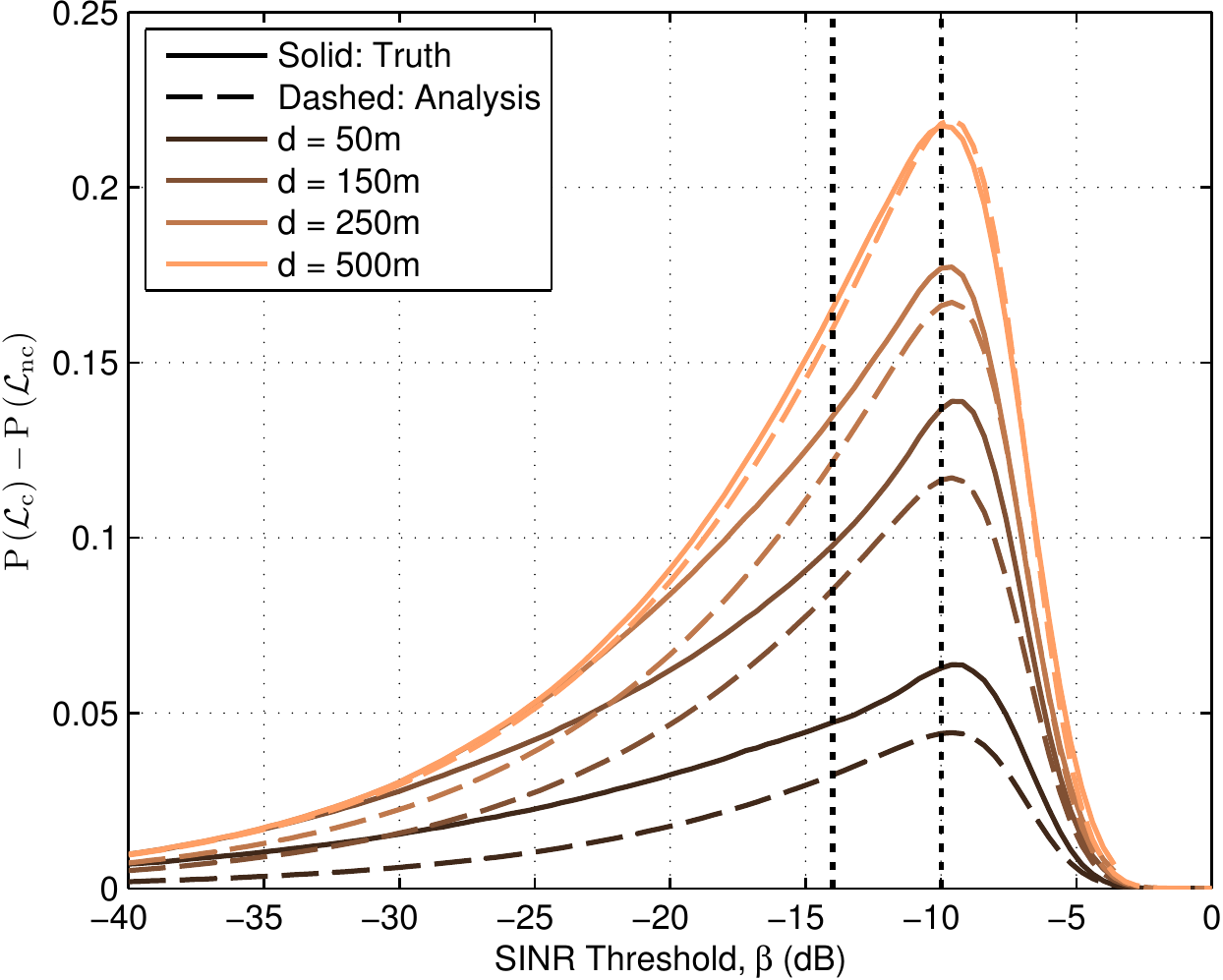}}
\caption{\textsc{Unique localizability} in range-based positioning (e.g., using RSS or TOA observations) for various collaborator separations $d$. The separation plays a major role in how beneficial collaboration will be. Note that the dotted lines at $\beta=-10$ and $-14$~dB delineate the range of SINR threshold values considered in 3GPP for far away BSs, which is just below where the peak collaborative benefits are obtained. ($\alpha=4$.)}
\label{Fig:P2_localizability_withoutshadowing}
\end{figure}

\begin{figure}
\centering
\subfloat[Probability of unique localizability]{\label{Fig:P2_celllocalizability_withoutshadowing_vs_d}\includegraphics[width=\figurewidth]{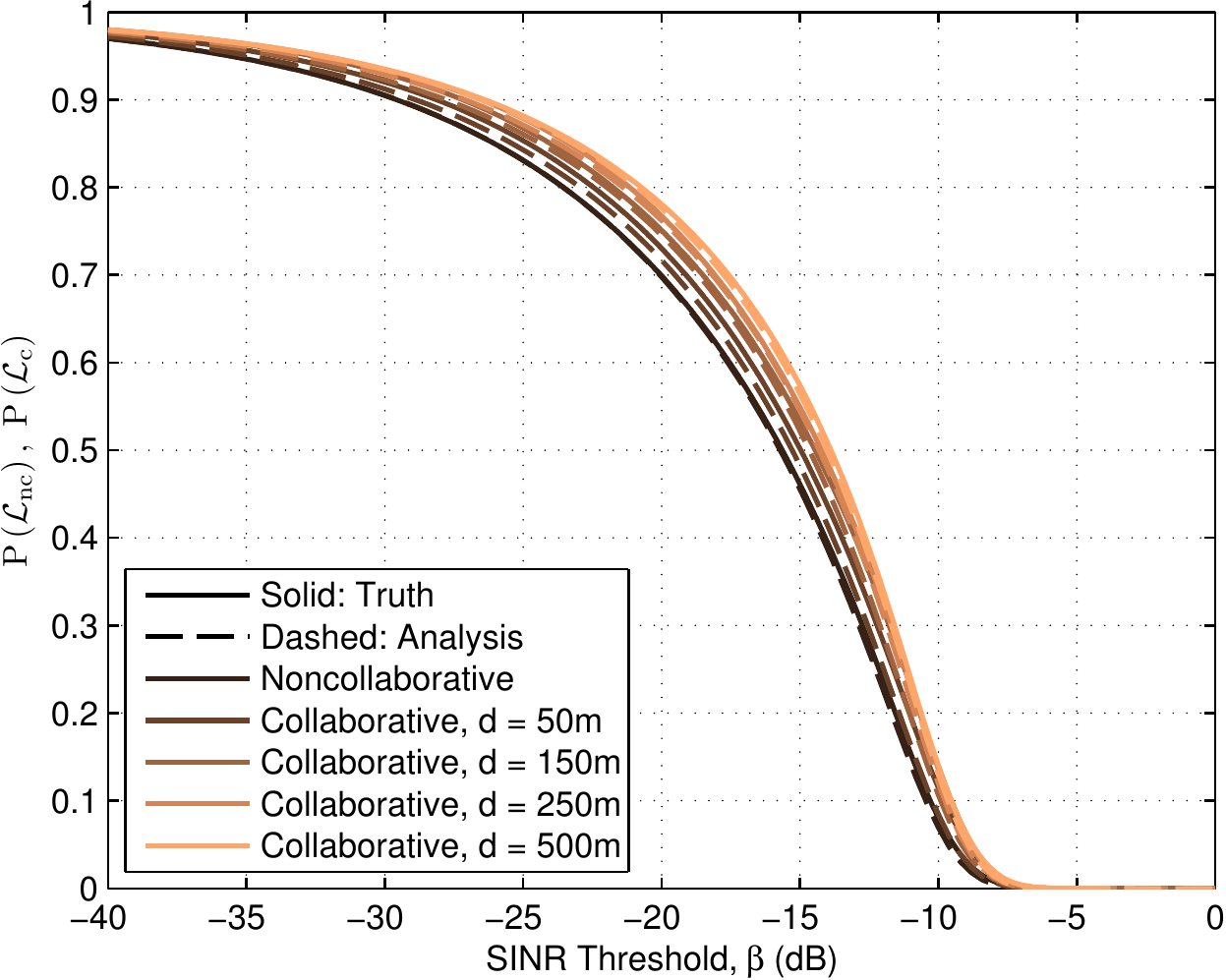}} \\
\subfloat[Absolute increase in unique localizability probability]{\label{Fig:P2_celllocalizability_improvement_withoutshadowing_vs_d}\includegraphics[width=\figurewidth]{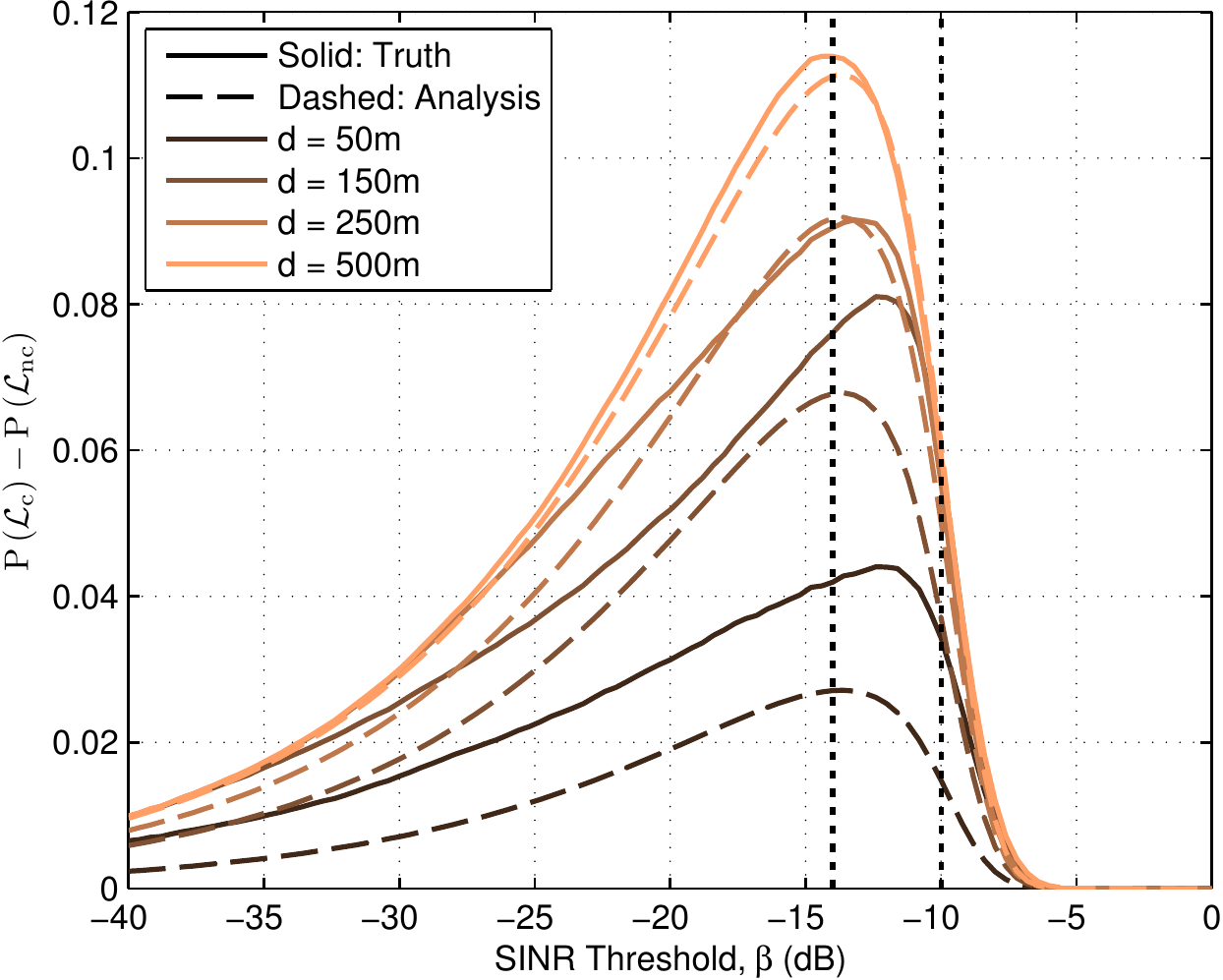}}
\caption{\textsc{Unique localizability} in range-difference-based positioning (e.g., using TDOA observations from BSs) for various collaborator separations $d$. The separation plays a major role in how beneficial collaboration will be. Note that the dotted lines at $\beta=-10$ and $-14$~dB delineate the range of SINR threshold values considered in 3GPP for far away BSs, which coincides with the peak collaborative benefits. ($\alpha=4$.)}
\label{Fig:P2_celllocalizability_withoutshadowing}
\end{figure}

\else

\begin{figure}
\centering
\subfloat[Probability of unique localizability]{\label{Fig:P2_localizability_withoutshadowing_vs_d}\includegraphics[width=\figurewidth]{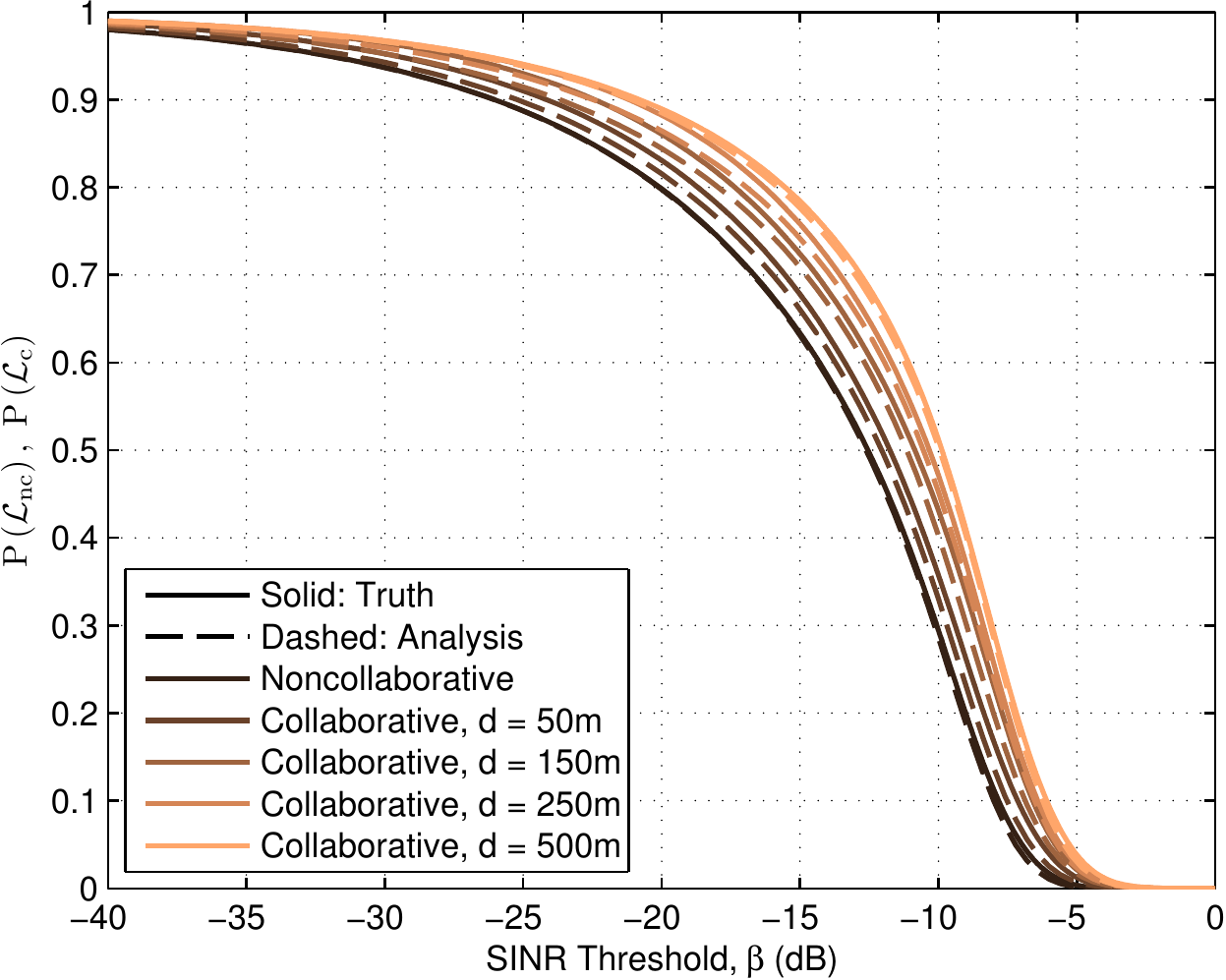}} \qquad
\subfloat[Absolute increase in unique localizability probability]{\label{Fig:P2_localizability_improvement_withoutshadowing_vs_d}\includegraphics[width=\figurewidth]{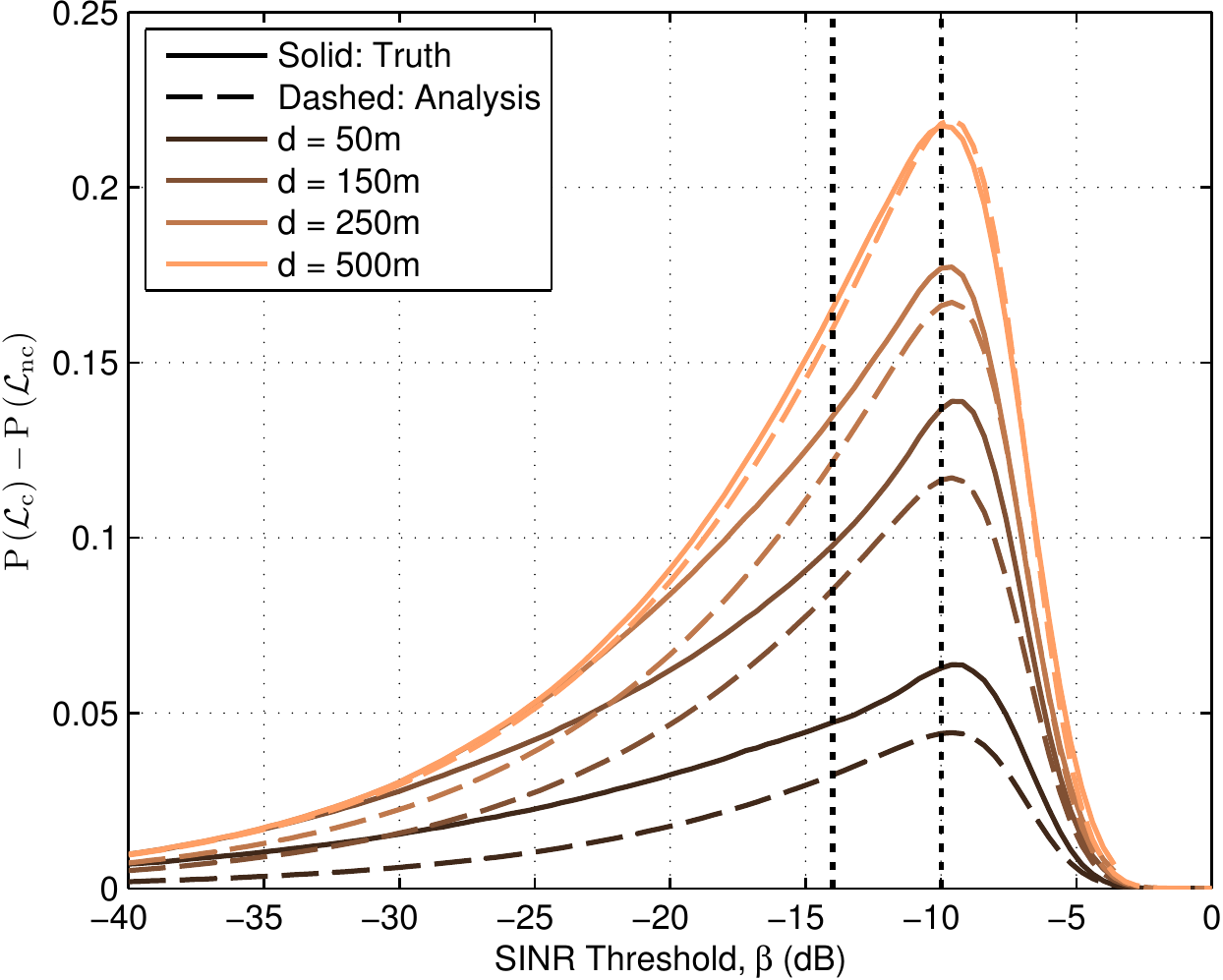}}
\caption{\textsc{Unique localizability} in range-based positioning (e.g., using RSS or TOA observations) for various collaborator separations $d$. The separation plays a major role in how beneficial collaboration will be. Note that the dotted lines at $\beta=-10$ and $-14$~dB delineate the range of SINR threshold values considered in 3GPP for far away BSs, which is just below where the peak collaborative benefits are obtained. ($\alpha=4$.)}
\label{Fig:P2_localizability_withoutshadowing}
\end{figure}

\begin{figure}
\centering
\subfloat[Probability of unique localizability]{\label{Fig:P2_celllocalizability_withoutshadowing_vs_d}\includegraphics[width=\figurewidth]{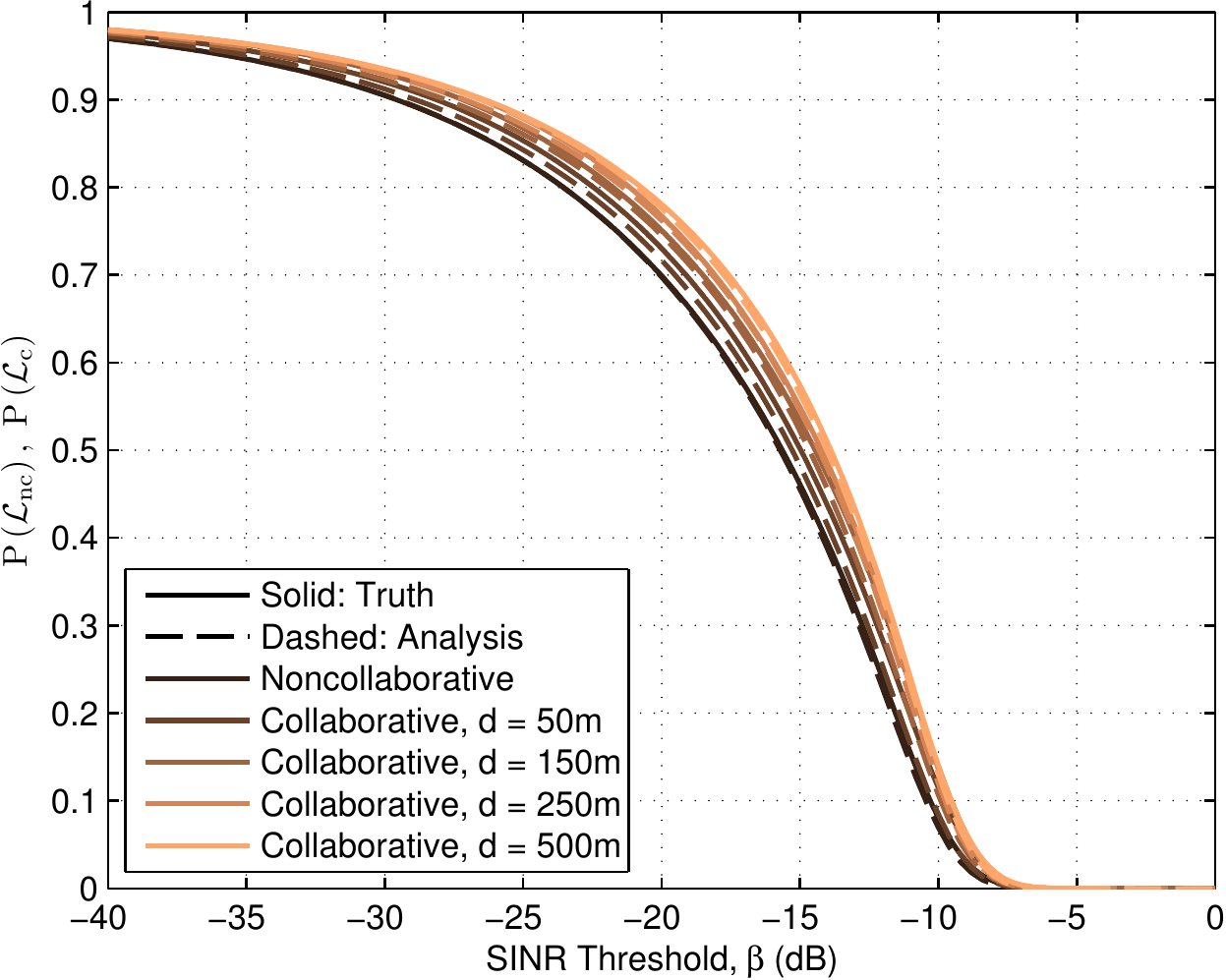}} \qquad
\subfloat[Absolute increase in unique localizability probability]{\label{Fig:P2_celllocalizability_improvement_withoutshadowing_vs_d}\includegraphics[width=\figurewidth]{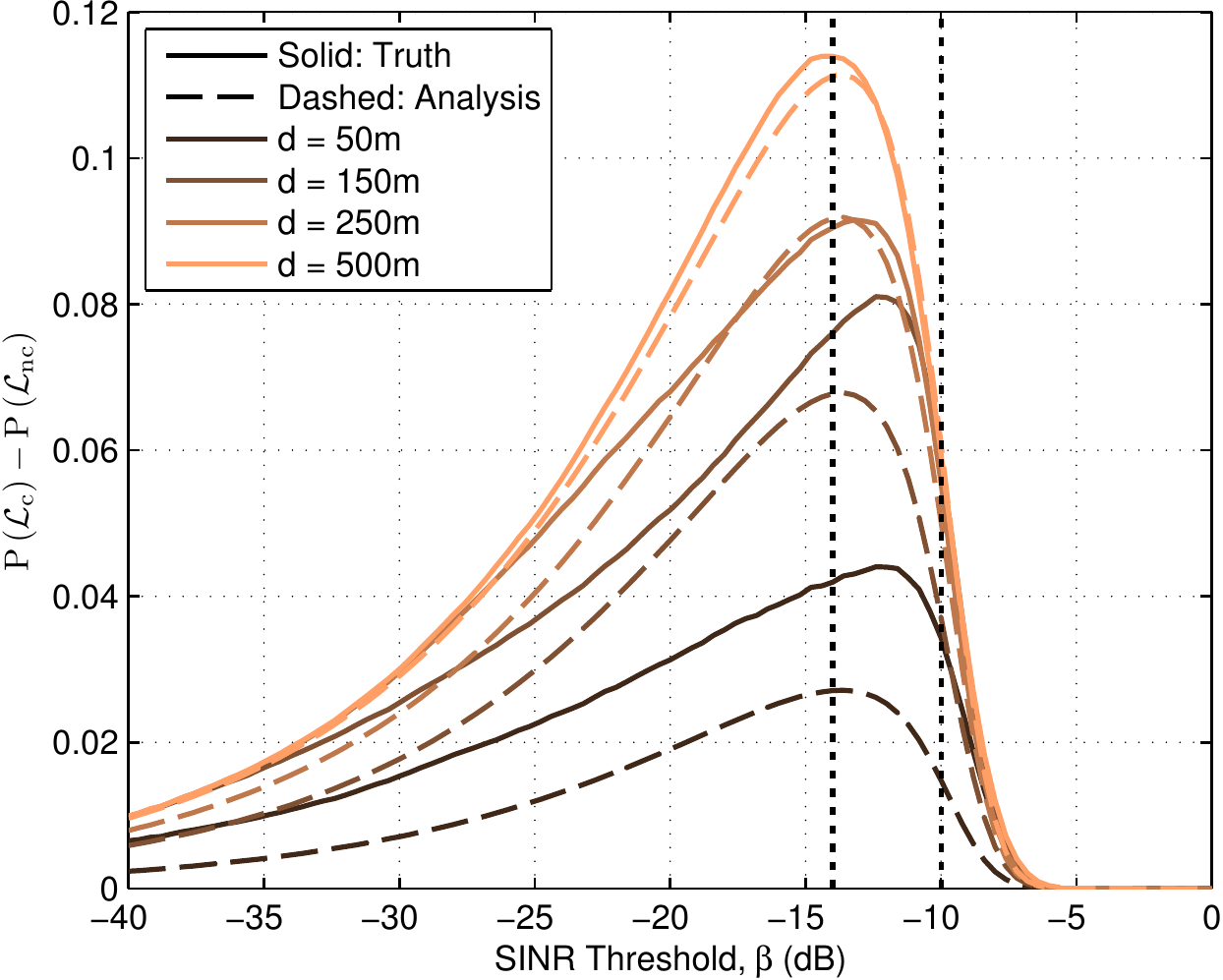}}
\caption{\textsc{Unique localizability} in range-difference-based positioning (e.g., using TDOA observations from BSs) for various collaborator separations $d$. The separation plays a major role in how beneficial collaboration will be. Note that the dotted lines at $\beta=-10$ and $-14$~dB delineate the range of SINR threshold values considered in 3GPP for far away BSs, which coincides with the peak collaborative benefits. ($\alpha=4$.)}
\label{Fig:P2_celllocalizability_withoutshadowing}
\end{figure}

\fi

Now, we consider the benefit of collaboration to localizability and how it is impacted by the SINR threshold $\threshold$. For $\alpha=4$, which we consider throughout as it is close to the 3.76 value used in 3GPP positioning studies~\cite{R1-091443} and allows the use of simplified expressions from~\cite{Schloemann2015c}, and various collaborator separations $d$, Figures~\ref{Fig:P2_localizability_withoutshadowing} and~\ref{Fig:P2_celllocalizability_withoutshadowing} show (a) the probability of unique localizability and (b) the absolute increase in these probabilities provided through device $\da$'s collaborative link with device $\db$ for range-based and range-difference-based localization, respectively. We note that for both $\ell=2$ and $\ell=3$, the approximation in~\eqref{Eq:PLc3} is within an absolute error of 0.02 from the truth, which is gathered via simulation. Furthermore, we see that the benefit from collaboration is a non-monotonic function of $\threshold$. For range-based positioning, $\threshold = -9$~dB appears to be a \emph{sweet spot} which maximizes the collaborative benefit, providing an approximately 6\% to 22\% absolute increase in the probabilities of unique localizability for the values of $d$ considered. For range-difference-based positioning, there appears to be a range of $\threshold$ values over which the benefit from collaboration is maximized, while overall, the benefit is reduced compared to that of range-based positioning. Interestingly, this range from approximately $\threshold = -10$ to $-14$~dB (highlighted in the figures with dotted lines) is right in line with the SINR thresholds discussed in 3GPP for the hearability of farther away BSs~\cite{R1-091912, Fischer2014}. These results bode well for the use of small-scale device-to-device collaborative ranging as a means to combat the hearability problem and improve the localizability of MDs in cellular networks. For example, when $\threshold=-12$~dB, an MD collaborating with another MD $d=150$m away observes an 8\% absolute increase in its unique localizability probability, a relative improvement of nearly 40\% from the 21\% in the noncollaborative case. Despite this significant relative improvement, an overall localizability probability of 29\% is still not acceptable for cellular positioning. Randomization due to the presence of shadowing, which is a more applicable scenario for cellular positioning, may actually help matters, as discussed next.

\subsection{The impact of shadowing on localizability}\label{Sec:P2:ImpactShadowing}

\ifx\sidebysidefigures\undefined

\begin{figure}
\centering
\subfloat[Range-based positioning]{\label{Fig:P2_localizability_withshadowing_vs_d}\includegraphics[width=\figurewidth]{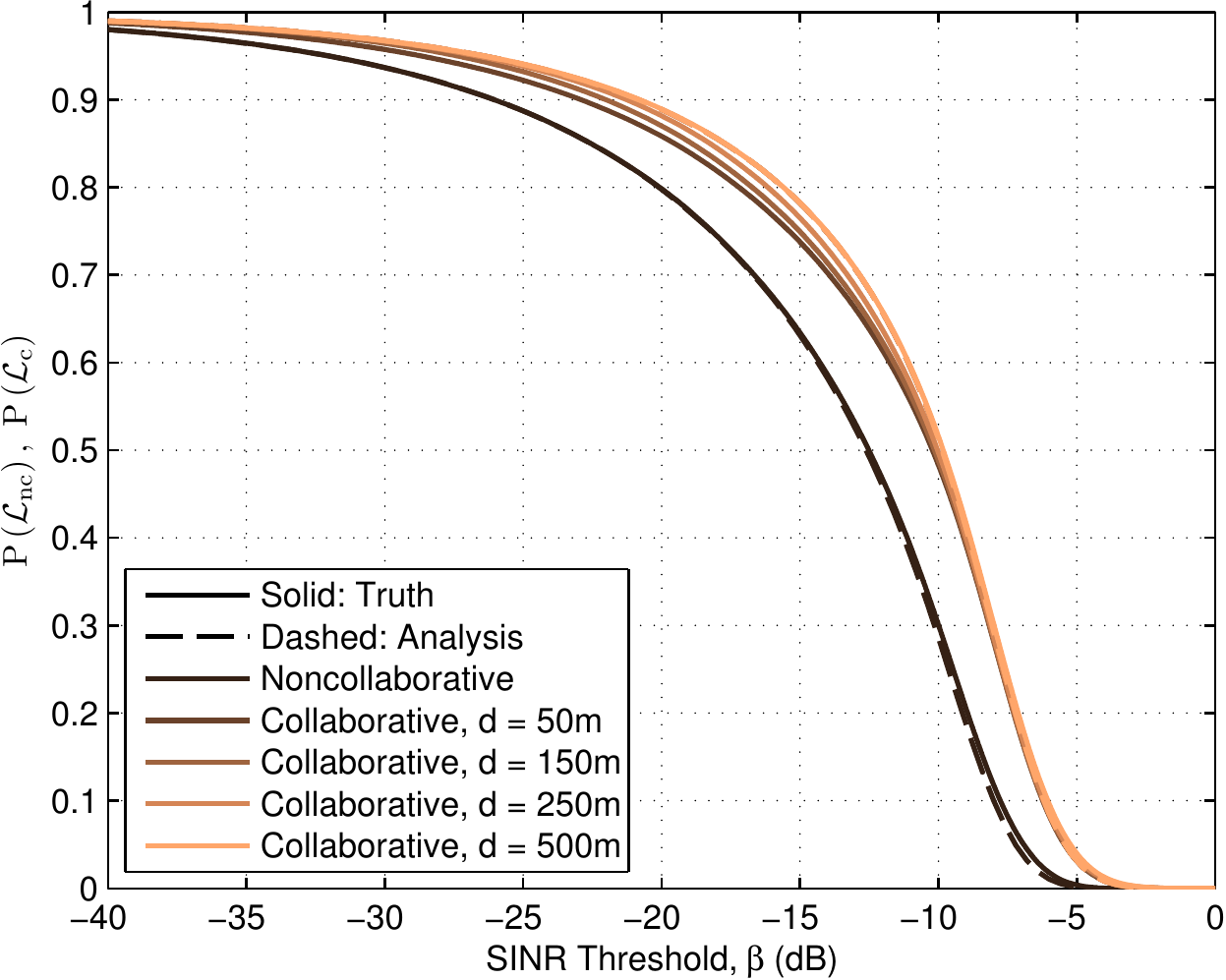}} \\
\subfloat[Range-difference-based positioning]{\label{Fig:P2_celllocalizability_withshadowing_vs_d}\includegraphics[width=\figurewidth,keepaspectratio]{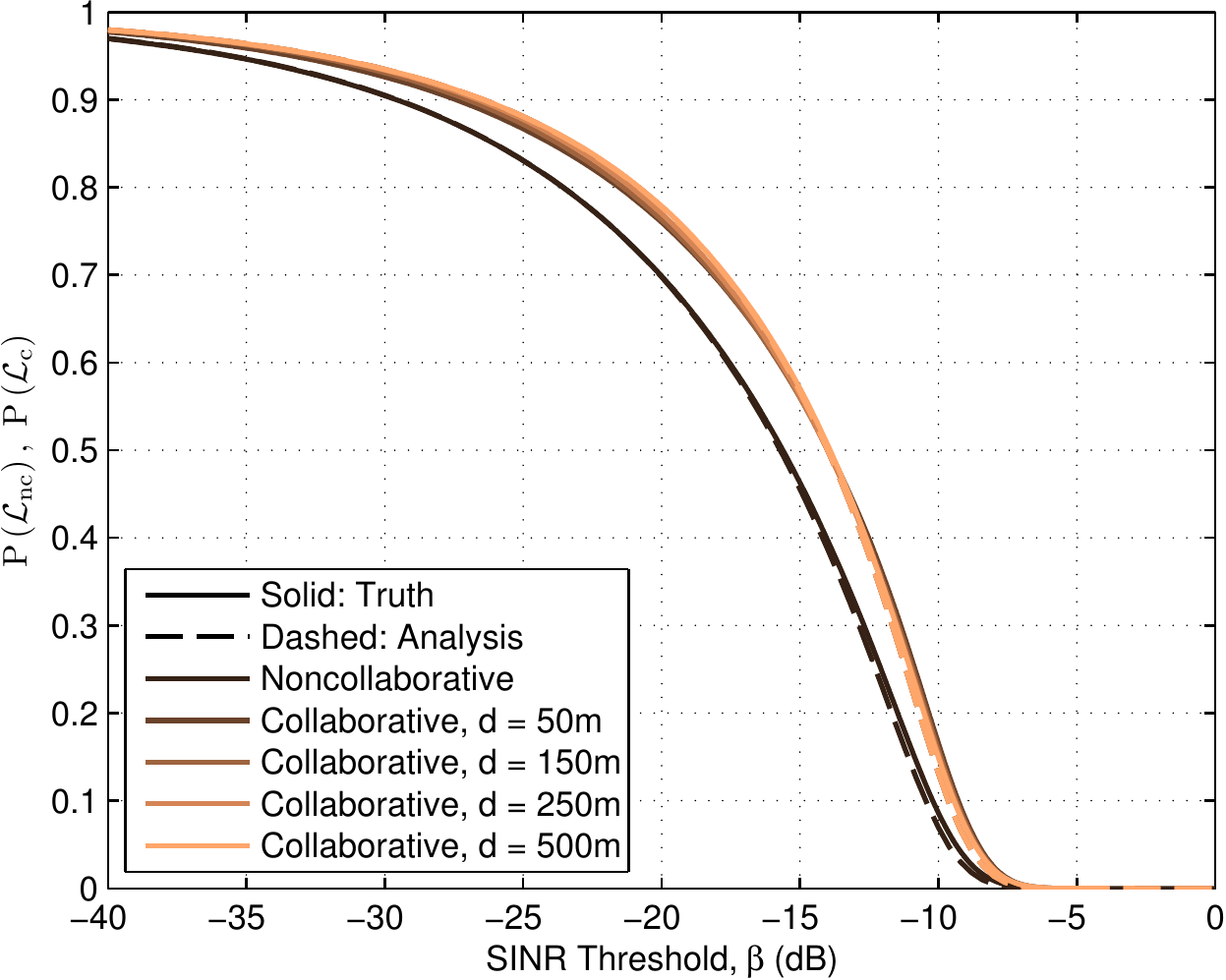}}
\caption{\textsc{The impact of shadowing} on the probability of unique localizability for various collaborator separations $d$. Due to the approximation in \eqref{Eq:P2:PLcShad}, the collaborative analysis lines overlap completely. Comparing with Figures~\ref{Fig:P2_localizability_withoutshadowing} and~\ref{Fig:P2_celllocalizability_withoutshadowing}, it is clear that shadowing significantly increases the benefit of collaboration, especially for smaller values of $d$. ($\alpha=4$.)}
\label{Fig:P2_withshadowing}
\end{figure}

\else

\begin{figure}
\centering
\subfloat[Range-based positioning]{\label{Fig:P2_localizability_withshadowing_vs_d}\includegraphics[width=\figurewidth]{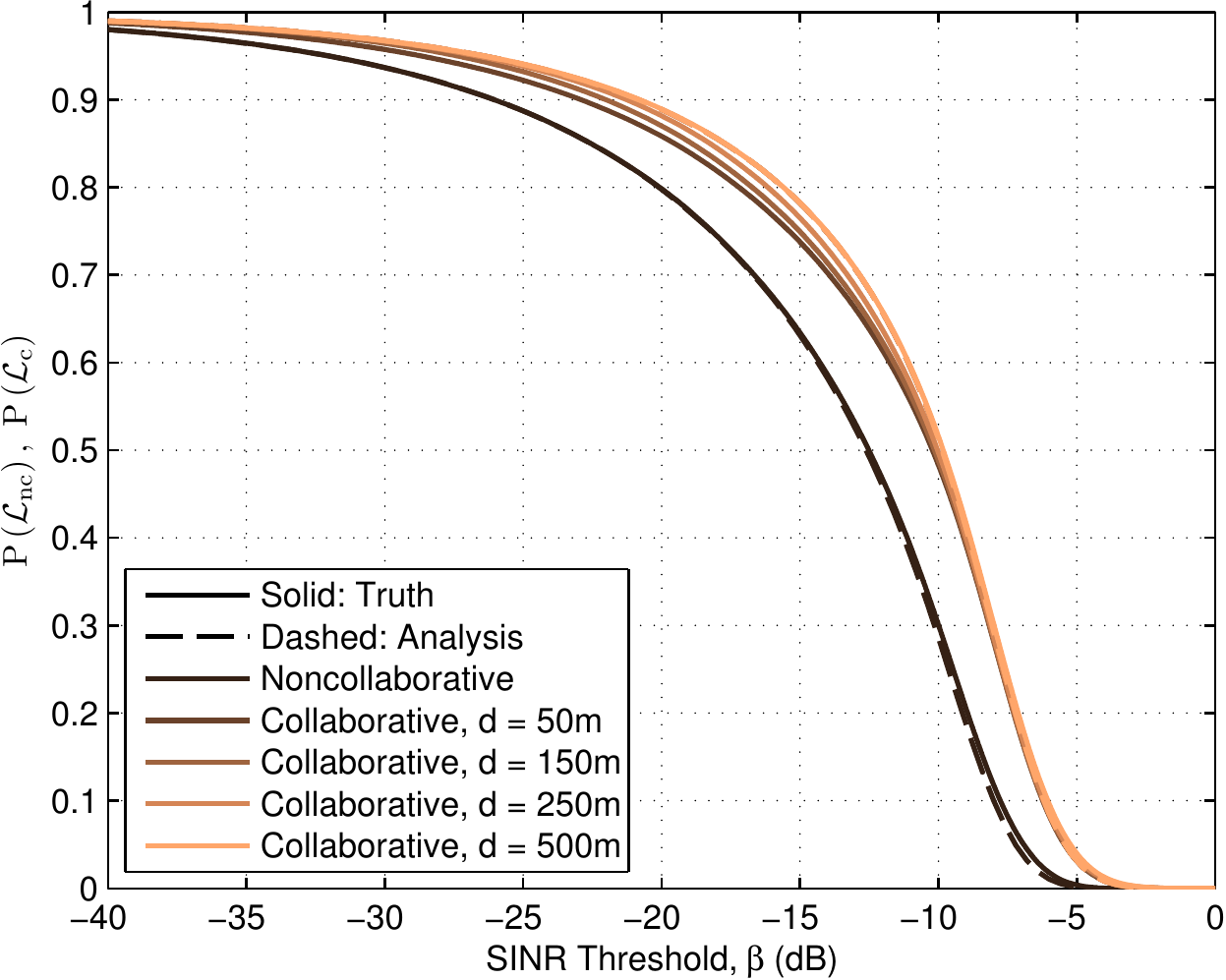}} \qquad
\subfloat[Range-difference-based positioning]{\label{Fig:P2_celllocalizability_withshadowing_vs_d}\includegraphics[width=\figurewidth]{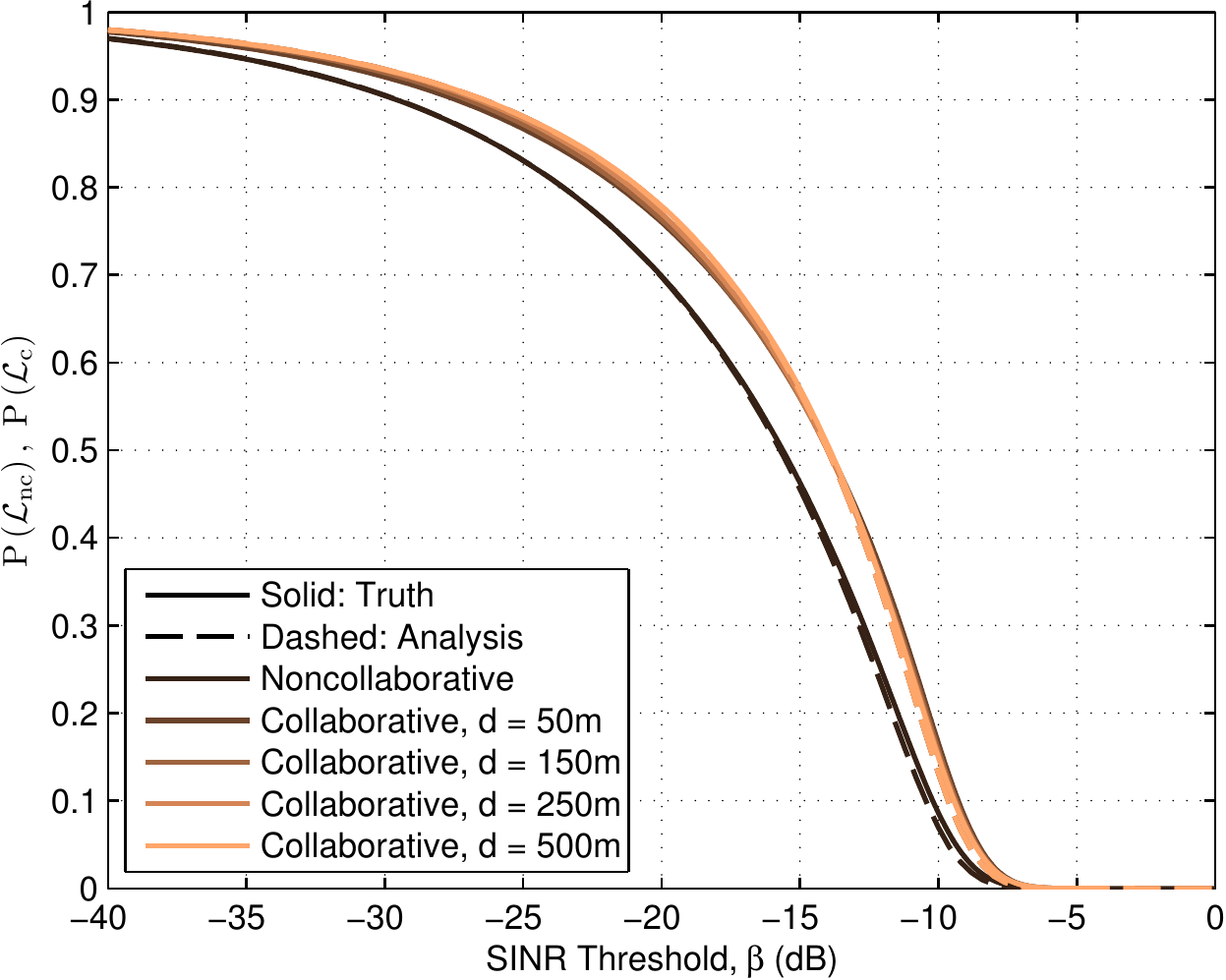}}
\caption{\textsc{The impact of shadowing} on the probability of unique localizability for various collaborator separations $d$. Due to the approximation in \eqref{Eq:P2:PLcShad}, the collaborative analysis lines overlap completely. Comparing with Figures~\ref{Fig:P2_localizability_withoutshadowing} and~\ref{Fig:P2_celllocalizability_withoutshadowing}, it is clear that shadowing significantly increases the benefit of collaboration, especially for smaller values of $d$. ($\alpha=4$.)}
\label{Fig:P2_withshadowing}
\end{figure}

\fi

\begin{figure}
\centering
\includegraphics[width=\figurewidth]{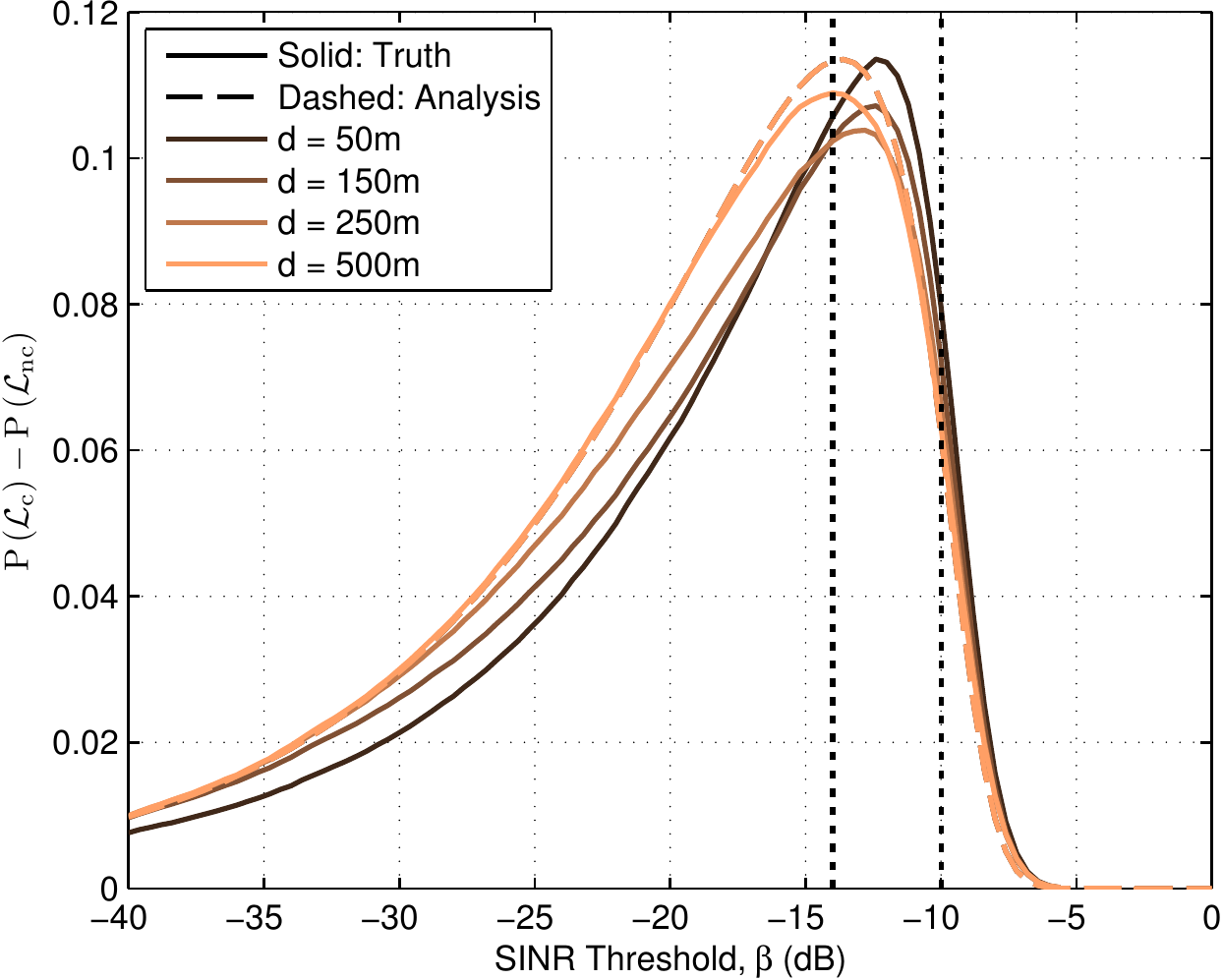}
\caption{\textsc{The improvement due to shadowing} on the probability of unique localizability for range-difference-based positioning. Here, the benefit from collaboration peaks at exactly the SINR thresholds which have been considered in 3GPP for far away BSs, between $\threshold=-10$ to $-14$~dB as highlighted by the dotted lines. ($\alpha=4$.)}
\label{Fig:P2_celllocalizability_improvement_withshadowing_vs_d}
\end{figure}

At this time, we consider the impact of shadowing using log-normal shadowing with $\sigma_s = 8$~dB and a correlation of 0.5 between the signals received at two devices originating from the same BS. Recall that~\eqref{Eq:P2:PLcShad} was derived using some very simplifying assumptions, including that two devices each hearing $\ell$ BSs certainly hear at least $\ell+1$ unique BSs, independent of the separation distance. Figure~\ref{Fig:P2_withshadowing} reveals that this was, in fact, not a bad assumption. For both range-based and range-difference-based positioning, the separation between the devices plays a highly-reduced role in the localizability probability compared to its role in the no shadowing case. Moreover, it is observed that the presence of shadowing is quite beneficial, primarily for increasing the probability that collaboration will improve localizability for shorter device separations $d$. This is more clearly observed in Figure~\ref{Fig:P2_celllocalizability_improvement_withshadowing_vs_d}, which plots exactly this improvement in localizability using range-difference observations from BSs for various different separations. For all values of $d$, the benefit from collaboration is very similar, which is quite different from the no shadowing case considered before, where the benefit clearly grew along with the collaborators' separation. Note that the peak localizability benefits, in the vicinity of 10\% to 11.5\%, are once again obtained between $\threshold = -10$ and $-14$~dB, the range of values commonly considered for $\threshold$ in 3GPP. Revisiting the example from the previous section ($\threshold=-12$~dB, $d=150$m), we see that shadowing has \emph{further} improved the probability of unique localizability in the collaborative case from 29\% to 32\%. While these results are promising and any improvement is welcome, it is clear that the additional collaborative link is not sufficient to provide truly-reliable localizability performance, at least not without frequency reuse, which we consider next.

\subsection{The impact of frequency reuse on cellular localizability}

\ifx\sidebysidefigures\undefined

\begin{figure}
\centering
\subfloat[Probability of unique localizability]{\label{Fig:P2_celllocalizability_withshadowing_andK_vs_d}\includegraphics[width=\figurewidth]{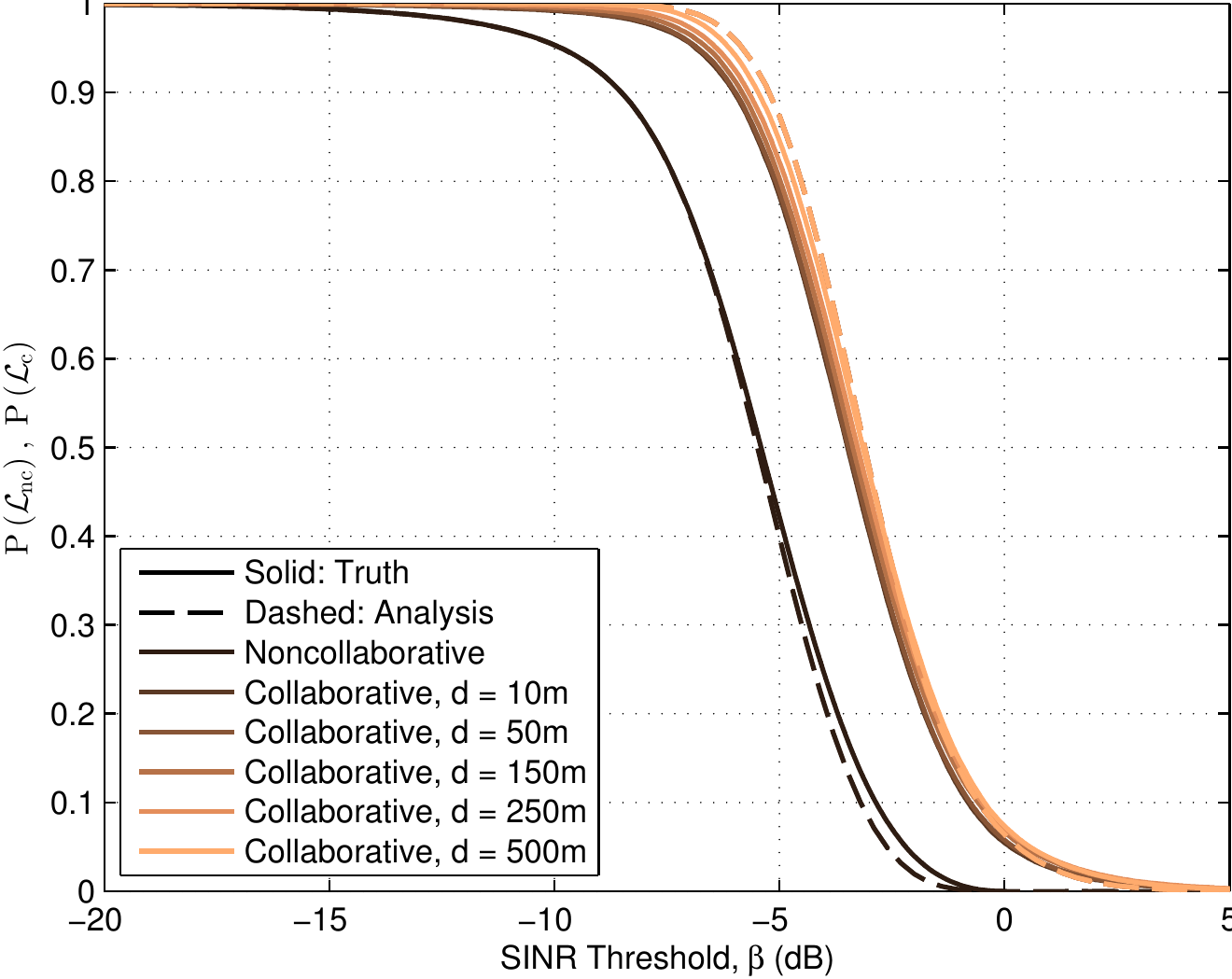}} \\
\subfloat[Absolute increase in unique localizability probability]{\label{Fig:P2_celllocalizability_improvement_withshadowing_andK_vs_d}\includegraphics[width=\figurewidth,keepaspectratio]{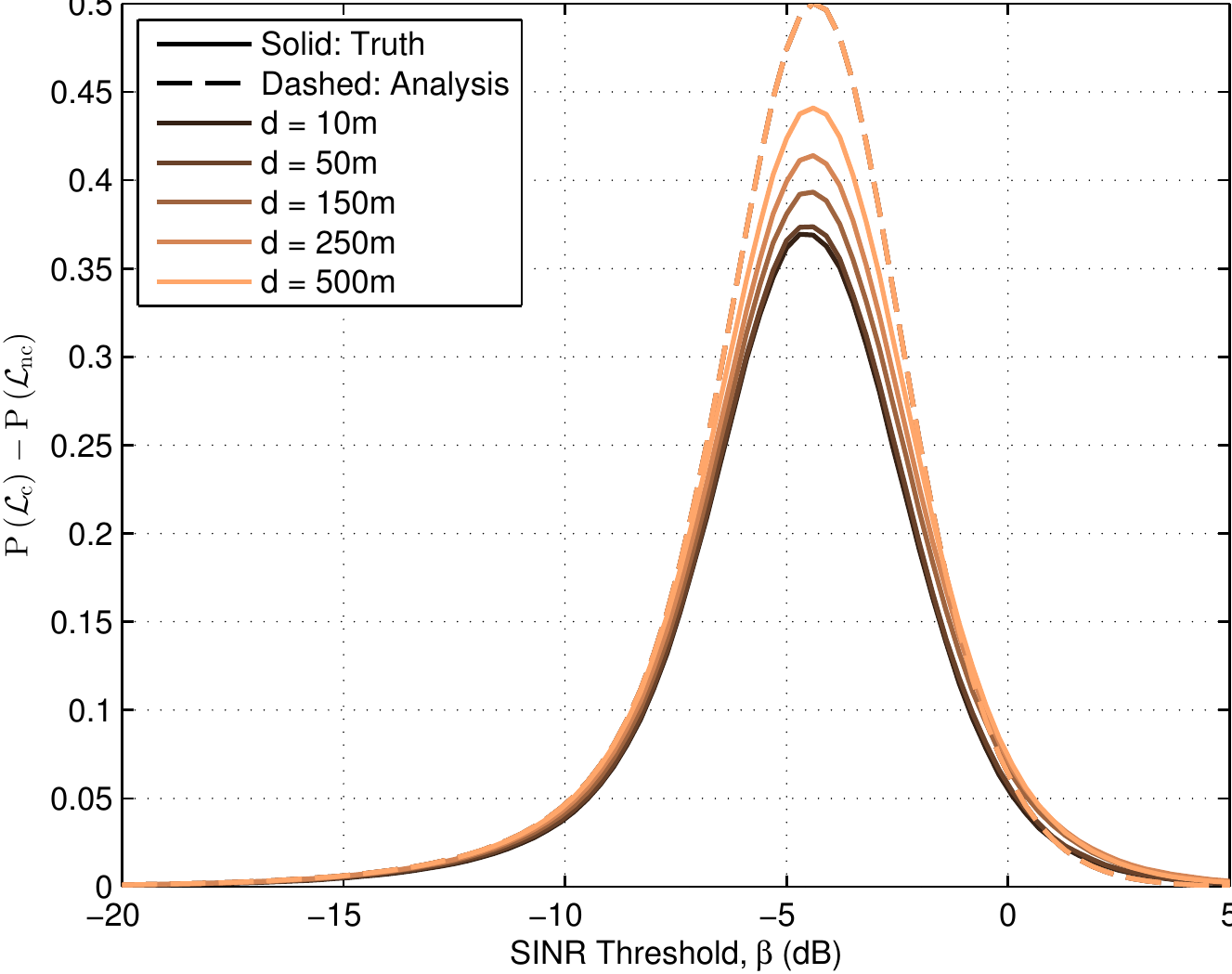}}
\caption{\textsc{Unique localizability} in range-difference-based positioning (e.g., using TDOA observations from BSs) for various collaborator separations $d$. The separation plays a major role in how beneficial collaboration will be. Note that the dotted lines at $\beta=-10$ and $-14$~dB delineate the range of SINR threshold values considered in 3GPP for far away BSs, which coincides with the peak collaborative benefits. ($\alpha=4$.)}
\label{Fig:P2_celllocalizability_withshadowing_andK}
\end{figure}

\else

\begin{figure}
\centering
\subfloat[Probability of unique localizability]{\label{Fig:P2_celllocalizability_withshadowing_andK_vs_d}\includegraphics[width=\figurewidth]{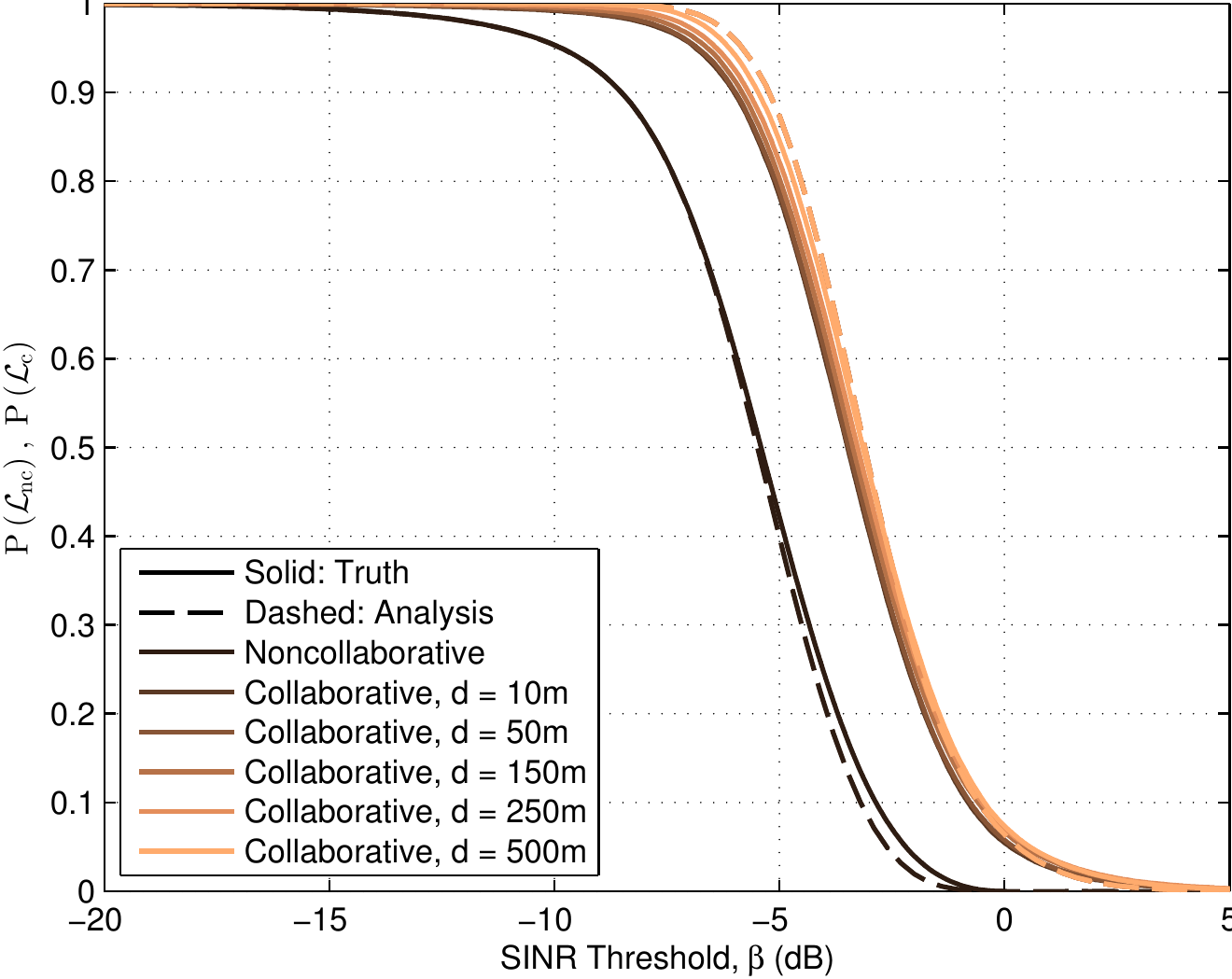}} \qquad
\subfloat[Absolute increase in unique localizability probability]{\label{Fig:P2_celllocalizability_improvement_withshadowing_andK_vs_d}\includegraphics[width=\figurewidth]{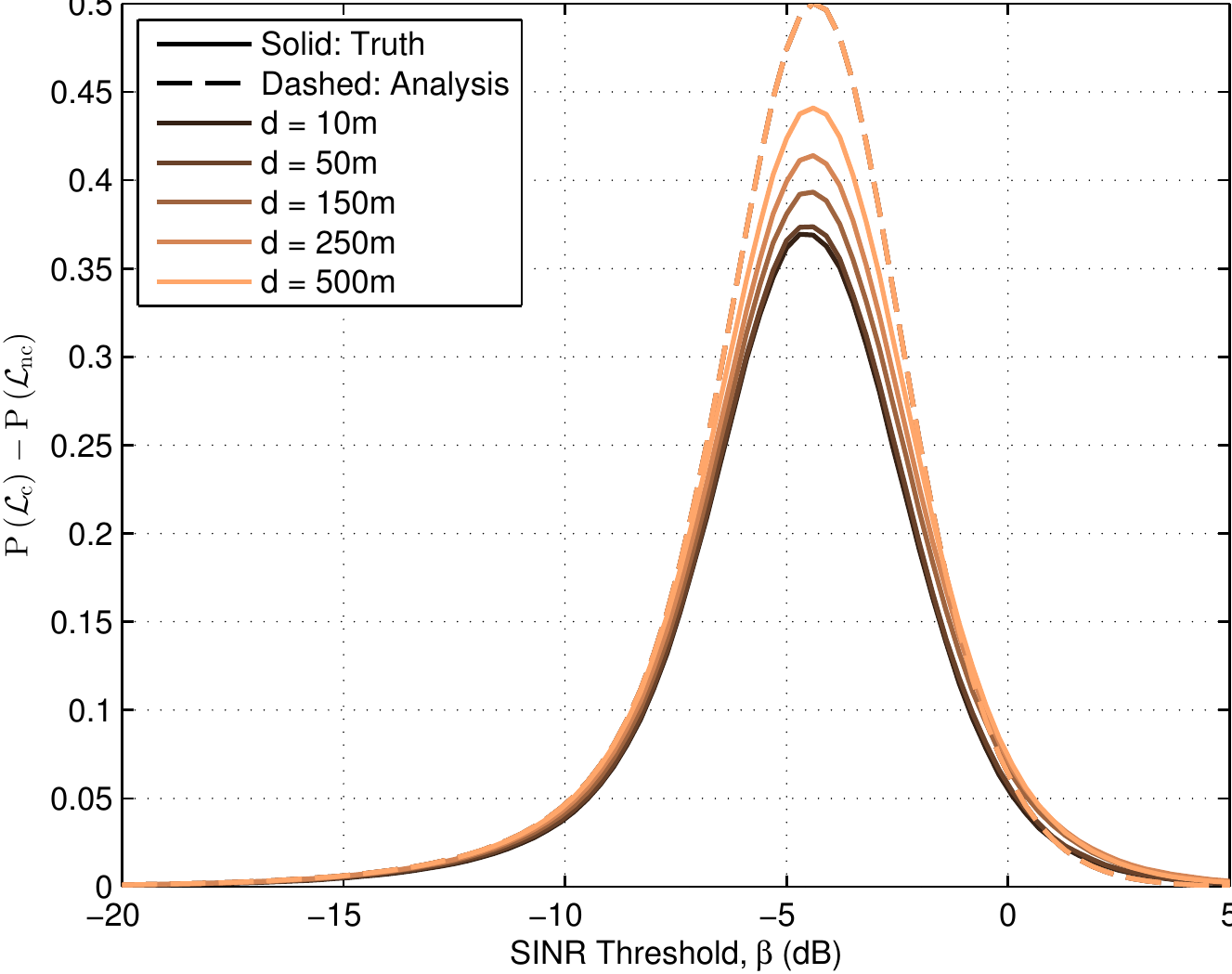}}
\caption{\textsc{The impact of frequency reuse} on unique localizability in range-difference-based positioning for various collaborator separations $d$ and a reuse factor of $K=3$. We note that frequency reuse greatly enhances the benefit of collaboration (even at very short $d=10$m communication ranges), which in turn may allow the employment of lower frequency reuse factors in cellular positioning. ($\alpha=4$.)}
\label{Fig:P2_celllocalizability_withshadowing_andK}
\end{figure}

\fi

To conclude our analysis, we lastly consider how \emph{frequency reuse}, commonly included in wireless standards~\cite{3GPP.36.305}, affects the value of collaboration for improving the probability of unique localizability. If a total of $K$ frequency bands are available and we independently assign one of the bands to each $x \in \PPPa$ with equal probability, we can easily incorporate frequency reuse into our model by considering the transmission activity on each band separately using independent BS PPPs whose densities are that of the original BS PPP thinned by the frequency reuse factor $K$. We consider a cellular positioning setup in which devices measure range-difference observations to BSs (e.g., OTDOA in LTE) and include correlated log-normal shadowing using the same parameter values as in Section~\ref{Sec:P2:ImpactShadowing}. If the number of hearable BSs in the $k\th$ band at some device $\nbz$ is $n_{\nbz}^k$, then the total number of BSs hearable at device $\nbz$ is $\nbbN_\nbz^{\tt K} = \sum_{k=1}^K n_{\nbz}^k$. Following the same logic as was employed in Section~\ref{Sec:P2:ShadowingCase} and letting $\ell=3$, we obtain $\P(\Lnc) = \P(\nbbN_\da^{\tt K} \geq 4)$ and $\P\left(\Lc\right) \approx \P\left(\nbbN_\da^{\tt K} \geq 4\right) + \P\left(\nbbN_\da^{\tt K} = 3\right)\P\left(\nbbN_\db^{\tt K} \geq 3\right)$ as the revised expressions of \eqref{Eq:P2:PLnc} and \eqref{Eq:PLc3} for this frequency reuse setup.\!\footnote{Using random frequency reuse, $\P(\nbbN_\nbz^{\tt K} \geq L) = \plk(1,1,\alpha,\threshold,1,\pppai)$ in Theorem~3 of~\cite{Schloemann2015c}.} For a frequency reuse factor of $K=3$, which pertains to OTDOA positioning using cell-specific reference signals (CRS) in LTE, the probabilities of unique localizability are plotted in Figure~\ref{Fig:P2_celllocalizability_withshadowing_andK}. While it is immediately clear that localizability has already improved drastically across the entire SINR range for noncollaborative positioning compared to its universal frequency reuse $K=1$ counterpart in Figure~\ref{Fig:P2_celllocalizability_withshadowing_vs_d}, collaboration adds an additional drastic improvement to the probability of unique localizability in this cellular scenario. In fact, collaboration appears to have a type of \emph{processing gain effect} allowing for $\threshold$ values between 2 to 5~dB higher than in noncollaborative positioning in order to achieve the same probability of unique localizability. Remarkably, this is true even when the collaborating devices are separated by only $d=10$m. Although the peak benefits are not in the $\beta = -10$ to $-14$~dB range anymore, we make a final remark here that the benefit from collaboration is sufficient to remove location ambiguities completely within this range. This could be very significant to cellular network operators, potentially making it possible to meet FCC E911 requirements using the CRS and fewer frequency bands than the current $K=6$ with positioning reference signals (PRS), which was deemed impossible for noncollaborative cellular positioning~\cite{R1-091912}.
%!TEX root = paper_2.tex

\section{Conclusion}

In this paper, we presented a tractable analysis of the impact of a single collaborative link on the probability of a mobile device being able to locate itself without ambiguity (i.e., being uniquely localizable). This is in contrast to previous works, which have primarily relied on simulations to study collaborative positioning in similar setups that include network self-interference. In the absence of shadowing and all other things being equal, the results show that collaboration is more beneficial to range-based than range-difference-based positioning systems. This makes logical sense since for two devices separated by some fixed distance, the sets of the two closest BSs to each device are less likely to be identical than the sets of the three closest BSs to each device, thus providing a greater likelihood of increasing the number of unique combined BSs. For both types of systems, it quickly becomes apparent that the key element affecting the value of collaboration is the separation between the devices; however, this is promptly deemphasized in the presence of shadowing. Although shadowing helps, the results make it clear that collaboration will not be sufficient to notably mitigate the hearability problem in cellular positioning systems employing universal frequency reuse. The localizability benefit from collaboration is greatly enhanced with frequency reuse, and it is seen that for a positioning system similar to that of OTDOA in LTE using CRS, a single collaborative link is sufficient to ensure unique localizability at the SINR detection thresholds commonly considered in 3GPP for cellular positioning (around -10~dB and below). These results are significant and demonstrate that short-distance small-scale collaboration, which is the most reasonable scenario for cellular networks, is a very worthwhile pursuit and likely to be a significant aid in the fight against the \emph{hearability problem}. In order to understand this further, we suggest that future work consider the exact localizability analysis in the shadowing case as well as present techniques for accurate device-to-device ranging.
%!TEX root = paper_2.tex

\appendix

\subsection{Proof of Lemma~\ref{Lemma:sameLth}}\label{Proof:sameLth}

\begin{figure}
\centering
\includegraphics[width=0.8\figurewidth]{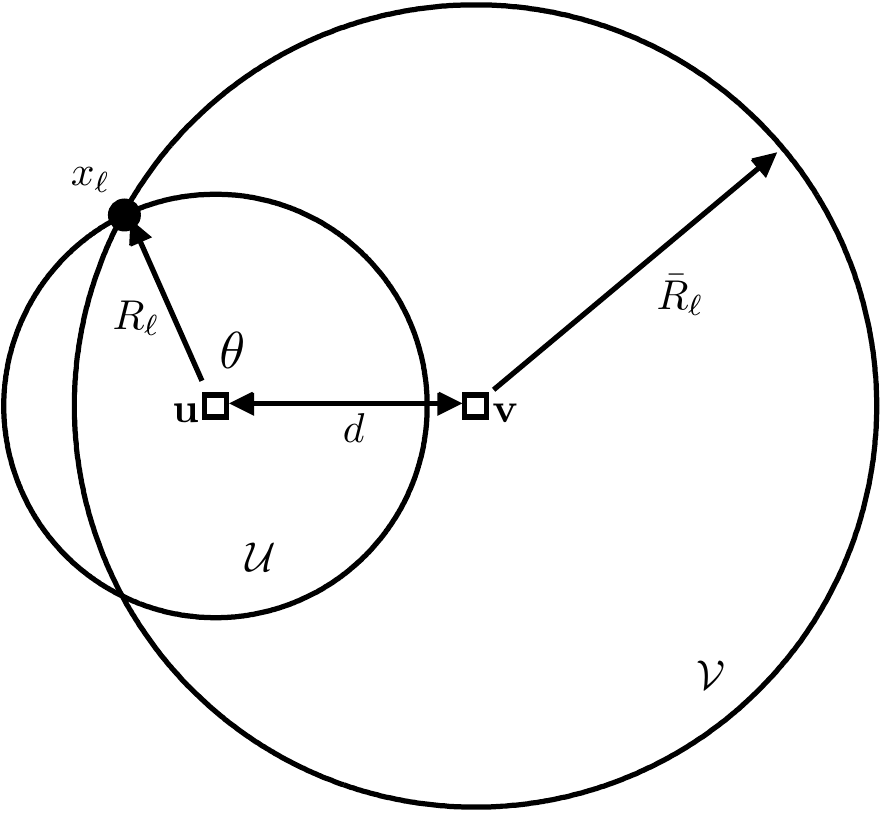}
\caption{\textsc{The setup of Lemma~\ref{Lemma:sameLth}}. This figure illustrates the setup used in the proof of Lemma~\ref{Lemma:sameLth} and highlights its key variables.}
\label{Fig:P2_sameLth}
\end{figure}

First, we define the variables used in the derivation, which are also highlighted in Figure~\ref{Fig:P2_sameLth}. As usual, let $\da$ and $\db$ represent two MDs separated by distance $d$. The location of the $\ell\th$ closest BS to $\da$ is $x_\ell$ and its distances to $\da$ and $\db$ are the random variables $R_\ell$ and $\bRL$, respectively. Next, $\ncalU$ is the region covered by the circle centered at $\da$ with radius $R_\ell$ and $\ncalV$ is the region covered by the circle centered at $\db$ with radius $\bRL$. Lastly, $\theta$ is a random variable which represents the angle of $x_\ell$ relative to the baseline connecting $\da$ and $\db$, i.e., $x_\ell = R_\ell \cdot [\cos \theta\ \sin \theta]^\T$. Now,
\begin{align*}
\P\left(\seta{\ell}\right. &= \left.\setb{\ell}, \seta{\ell-1} = \setb{\ell-1}  \middle\vert  D=d\right) \notag\\
&\sreq{a} \E_{R_\ell}\left[\E_\theta\left[\P\left(\seta{\ell} = \setb{\ell}, \seta{\ell-1} = \setb{\ell-1}  \middle\vert  d, R_\ell, \theta\right)\right]\right] \notag\\
&\sreq{b} \E_{R_\ell}\left[\E_\theta\left[\P\left(\ncalN_{\ncalU \cap \ncalV} = \ell, \ncalN_{\ncalV \backslash (\ncalU \cap \ncalV)} = 0  \middle\vert  d, R_\ell, \theta\right)\right]\right] \notag\\
&\sreq{c} \E_{R_\ell}\left[\E_\theta\left[\P\left(\ncalN_{\ncalU \cap \ncalV} = \ell  \middle\vert  d, R_\ell, \theta\right) \P\left(\ncalN_{\ncalV \backslash (\ncalU \cap \ncalV)} = 0  \middle\vert  d, R_\ell, \theta\right)\right]\right] \notag\\
&\sreq{d} \E_{R_\ell}\left[\E_\theta\left[\left(\frac{\vert\ncalU \cap \ncalV\vert}{\vert\ncalU\vert}\right)^{\ell-1} e^{-\pppai\vert\ncalV \backslash (\ncalU \cap \ncalV)\vert} \right]\right] \notag\\
&= \int_0^\infty \int_0^{2\pi} \left(\frac{\vert\ncalU \cap \ncalV\vert}{\vert\ncalU\vert}\right)^{\ell-1} e^{-\pppai\vert\ncalV \backslash (\ncalU \cap \ncalV)\vert} f_{\theta}(\left.\theta \middle\vert r\right) f_{R_\ell}(r;\ell,\pppai)\ \d\theta\ \d{r} \notag\\
&\sreq{e} \frac{1}{\pi} \int_0^\infty \int_0^{2\pi} \left(\frac{\vert\ncalU \cap \ncalV\vert}{\vert\ncalU\vert}\right)^{\ell-1} e^{-\pppai(\vert\ncalV \backslash (\ncalU \cap \ncalV)\vert+\pi r^2)} \frac{(\pppai \pi r^2)^\ell}{r \Gammal}\ \d\theta\ \d{r} \notag\\
&\sreq{f} \frac{1}{\pi} \int_0^\infty \int_0^{2\pi} \left(\frac{\pi r^2-\Alune(r, \bL, d)}{\pi r^2}\right)^{\ell-1} \notag \\
&\qquad \times e^{-\pppai(\Alune(\bL, r, d)+\pi r^2)} \frac{(\pppai \pi r^2)^\ell}{r \Gammal}\ \d\theta\ \d{r} \notag\\
&\sreq{g} \frac{2}{\pi\Gammal} \int_0^\infty \int_0^{\pi} \left(\frac{\pi r^2-\Alune(r, \bL, d)}{\pi r^2}\right)^{\ell-1} \notag \\
&\qquad \times e^{-\pppai(\Alune(\bL, r, d)+\pi r^2)} \frac{(\pppai \pi r^2)^\ell}{r}\ \d\theta\ \d{r},
\end{align*}
where $(a)$ follows from fixing the location of $x_\ell$, i.e., fixing both $R_\ell$ and $\theta$ and taking the expectation over their distributions, $(b)$ follows from the fact that, conditioned on $x_{\ell}$, for $\seta{\ell} = \setb{\ell}$ and $\seta{\ell-1} = \setb{\ell-1}$, there must be total of $\ell$ BSs in both $\ncalU \cap \ncalV$ and $\ncalU \cup \ncalV$, $(c)$ follows from the independence of regions $\ncalU \cap \ncalV$ and $\ncalV \backslash (\ncalU \cap \ncalV)$, $(d)$ follows from calculating the probability that the $\ell-1$ BSs \emph{known to be inside} $\ncalU$ are all in $\ncalU \cap \ncalV$ and calculating the void probability of $\ncalV \backslash (\ncalU \cap \ncalV)$, 
$(e)$ follows from $f_\theta\left(\theta \middle\vert r\right) = f_\theta\left(\theta\right) = 1/2\pi$ and, just like \eqref{Eq:P2:D_K},
\begin{align}
	f_{R_\ell}(r; \ell, \pppai)
	&= e^{-\pppai \pi r^2 } \frac{2 (\pppai \pi r^2)^\ell}{r \Gammal}, \label{Eq:P2:R_ell}
\end{align}
$(f)$ follows from
\begin{align}
\vert\ncalU\vert &= \pi R_{\ell}^2, \label{Eq:AreaA}\\
\vert\ncalU \cap \ncalV\vert &= \pi R_{\ell}^2-\Alune(R_{\ell}, \bRL, d), \label{Eq:AreaAintB}\\
\vert\ncalV \backslash (\ncalU \cap \ncalV)\vert &= \Alune(\bRL, R_{\ell}, d), \label{Eq:AreaBnotAintB}
\end{align}
where $R_\ell = r$ and $\bRL = \bL$ here, and $(g)$ follows from multiplying by 2 and halving the integration limits of $\theta$ due to symmetry.

\subsection{Proof of Lemma~\ref{Lemma:diffLth}}\label{Proof:diffLth}

\begin{figure}
\centering
\includegraphics[width=0.8\figurewidth]{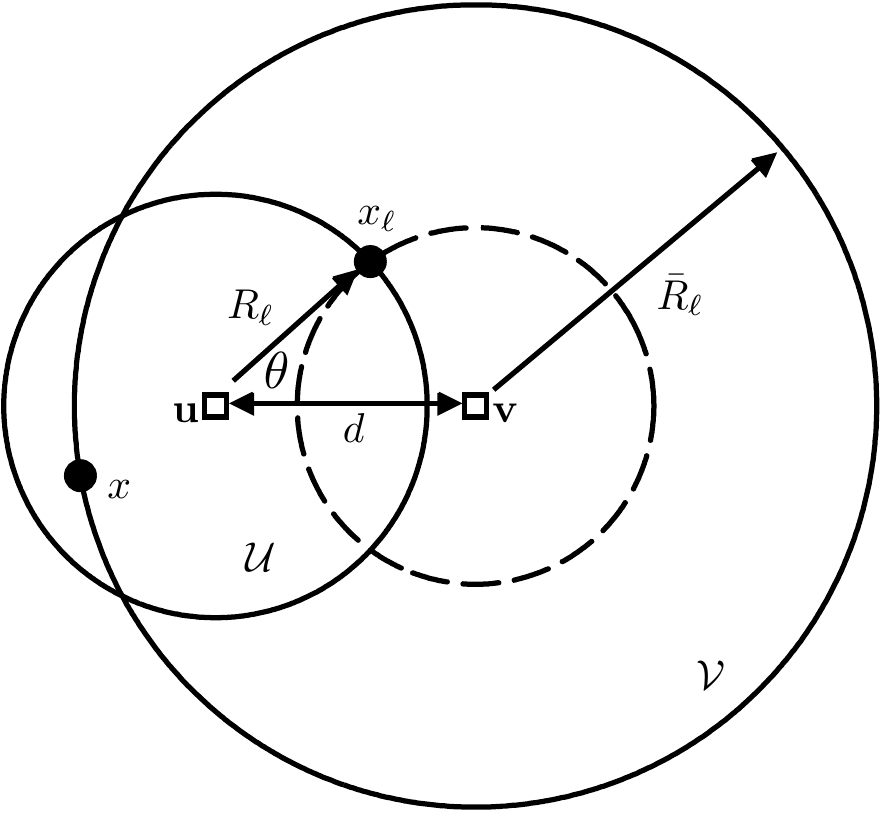}
\caption{\textsc{The setup of Lemma~\ref{Lemma:diffLth}}. This figure illustrates the setup used in the proof of Lemma~\ref{Lemma:diffLth} and highlights its key variables.}
\label{Fig:P2_diffLth}
\end{figure}

Again, we first define the variables used in the derivation, which are also highlighted in Figure~\ref{Fig:P2_diffLth}. As before, let $\da$ and $\db$ represent two MDs separated by distance $d$. The location of the $\ell\th$ closest BS to $\da$ is $x_\ell$ and its distance to $\da$ is $R_\ell$. Let $\nbX = \{\ptx_1, \ldots, \ptx_{\ell-1}\}$ represent the unordered set of $\ell-1$ closest BSs to $\da$ and $\setdptxtob = \{\dptxtob_1, \ldots, \dptxtob_{\ell-1}\}$ their corresponding distances to $\db$. Let $x \in \nbX$ be the location of the $\ell\th$ closest BS to $\db$ and $\bRL \in \setdptxtob$ the corresponding distance between the two. Next, $\ncalU$ is the region covered by the circle centered at $\da$ with radius $R_\ell$ and $\ncalV$ is the region covered by the circle centered at $\db$ with radius $\bRL$. The random variable $\theta$ represents the angle of $x_\ell$ relative to the baseline connecting $\da$ and $\db$. Now,
\begin{align*}
\P\left(\seta{\ell}\right. &= \left.\setb{\ell}, \seta{\ell-1} \neq \setb{\ell-1}  \middle\vert  D=d\right) \notag\\
&\sreq{a} \E_{R_{\ell}}\left[\E_\theta\left[\P\left(\seta{\ell} = \setb{\ell}, \seta{\ell-1} \neq \setb{\ell-1}  \middle\vert  d, R_{\ell}, \theta\right)\right]\right] \notag\\
&\sreq{b} \E_{R_{\ell}}\left[\E_\theta\left[ \sum_{i=1}^{\ell-1} \P\left(\dptxtob_i = \max\{\setdptxtob\}, \ncalN_\ncalV = \ell  \middle\vert  d, R_{\ell}, \theta\right)\right]\right] \notag\\
&\sreq{c} \E_{R_{\ell}}\left[\E_\theta\left[ \sum_{i=1}^{\ell-1} \E_{\dptxtob_i}\left[\E_{\phi_i}\left[ \P\left(\dptxtob_i = \max\{\setdptxtob\}, \ncalN_\ncalV = \ell  \middle\vert  d, R_{\ell}, \theta, \dptxtob_i, \phi_i\right)\right]\right]\right]\right] \notag\\
&\sreq{d} \E_{R_{\ell}}\left[\E_\theta\left[ \sum_{i=1}^{\ell-1}
\int_{\bLx}^{d+R_{\ell}} \int_{\pi - \cos^{-1}\!\left(\frac{d^2 + x^2 - \ell^2}{2 \cdot d \cdot x}\right)}^{\pi + \cos^{-1}\!\left(\frac{d^2 + x^2 - \ell^2}{2 \cdot d \cdot x}\right)} \left( \frac{\vert\ncalU \cap \ncalV\vert}{\vert\ncalU\vert} \right)^{\ell-2} \right.\right.\notag \\
&\qquad \times \left.\left. e^{-\pppai \vert\ncalV \backslash (\ncalU \cap \ncalV)\vert} f_\phi\left(\phi \middle\vert d,R_{\ell},x\right)f_x\left(x \middle\vert d,R_{\ell}\right)\ \d \phi\ \d x\ \right]\right] \notag\\
&\sreq{e} \E_{R_{\ell}}\left[\E_\theta\left[ \sum_{i=1}^{\ell-1}
\int_{\bLx}^{d+R_{\ell}} \left( \frac{\vert\ncalU \cap \ncalV\vert}{\vert\ncalU\vert} \right)^{\ell-2} e^{-\pppai \vert\ncalV \backslash (\ncalU \cap \ncalV)\vert} f_x\left(x \middle\vert d,R_{\ell}\right)\  \d x\ \right]\right] \notag\\
%&\sreq{f} (\ell-1) \E_{R_{\ell}}\left[\E_\theta\left[\int_{\bLx}^{d+R_{\ell}} \left( \frac{\vert\ncalU \cap \ncalV\vert}{\vert\ncalU\vert} \right)^{\ell-2} \right.\right.\notag \\
%&\qquad \times \left.\left. e^{-\pppai \vert\ncalV \backslash (\ncalU \cap \ncalV)\vert} \frac{\phirange(d, R_\ell, x)\,x}{\vert\ncalU\vert}\  \d x\ \right]\right] \notag\\
&\sreq{f} (\ell-1) \E_{R_{\ell}}\left[\E_\theta\left[\int_{\bLx}^{d+R_{\ell}} \left( \frac{\vert\ncalU \cap \ncalV\vert}{\vert\ncalU\vert} \right)^{\ell-2}\!\!\!e^{-\pppai \vert\ncalV \backslash (\ncalU \cap \ncalV)\vert} \frac{\phirange(d, R_\ell, x)\,x}{\vert\ncalU\vert}\  \d x\ \right]\right] \notag\\
&\sreq{g} \frac{\ell-1}{2\pi} \int_0^\infty \int_0^{2\pi} \int_{\bL}^{d+r} \left( \frac{\vert\ncalU \cap \ncalV\vert}{\vert\ncalU\vert} \right)^{\ell-2} \notag \\
&\qquad \times e^{-\pppai \vert\ncalV \backslash (\ncalU \cap \ncalV)\vert} \frac{\phirange(d,r,x)\,x}{\vert\ncalU\vert} f_{R_{\ell}}(r;\ell,\pppai)\  \d x\ \d \theta\ \d r \notag\\
&\sreq{h} \frac{\ell-1}{\pi} \int_0^\infty \frac{1}{\vert\ncalU\vert} \int_0^{\pi} \int_{\bL}^{d+r} \left( \frac{\vert\ncalU \cap \ncalV\vert}{\vert\ncalU\vert} \right)^{\ell-2} \notag \\
&\qquad \times e^{-\pppai \vert\ncalV \backslash (\ncalU \cap \ncalV)\vert} \phirange(d,r,x)\,x\, f_{R_{\ell}}(r;\ell,\pppai)\  \d x\ \d \theta\ \d r \notag\\
&\sreq{i} \frac{\ell-1}{\pi} \int_0^\infty \frac{1}{\pi r^2} \int_0^{\pi} \int_{\bL}^{d+r} \left( \frac{\pi r^2-\Alune(r,x,d)}{\pi r^2} \right)^{\ell-2} \notag \\
&\qquad \times e^{-\pppai \Alune(x,r,d)} \phirange(d,r,x)\,x\, f_{R_{\ell}}(r;\ell,\pppai)\  \d x\ \d \theta\ \d r, \notag\\
&\sreq{j} \frac{2(\ell-1)}{\pi\Gammal} \int_0^\infty \frac{1}{\pi r^2} \int_0^{\pi} \int_{\bL}^{d+r} \left( \frac{\pi r^2-\Alune(r,x,d)}{\pi r^2} \right)^{\ell-2} \notag \\
&\qquad \times e^{-\pppai (\Alune(x,r,d)+\pi r^2)} \phirange(d,r,x)\frac{x(\pppai \pi r^2)^\ell}{r}\  \d x\ \d \theta\ \d r,
\end{align*}
where $(a)$ follows from fixing the location of $x_\ell$, i.e., fixing both $R_\ell$ and $\theta$ and taking the expectation over their distributions, $(b)$ follows from the law of total probability by summing over the probabilities that each of the $\ell-1$ points in the interior of $\ncalU$ is the furthest point from $\db$ and that there are exactly $\ell$ points in the resulting region $\ncalV$, $(c)$ follows from fixing the location of $\ptx_i = \db + \dptxtob_i \cdot [\sin \phi_i\ \cos \phi_i]^\T$ (similarly to $(a)$) and taking the expectation over its distribution, $(d)$ follows from rewriting the two inner expectations as integrals and adjusting the integration bounds to remove values which lead to trivial zero-valued integrands, $(e)$ follows from the fact that the integration over $\phi$ is nothing less than a complete integration of its density over its entire support, $(f)$ follows from the fact that the summand is independent of the index $i$ and from substituting in the conditional density of $x$, where
\[
\phirange(d,R_{\ell},x) = 2 \cos^{-1}\!\left(\frac{d^2 + x^2 - R_{\ell}^2}{2 \cdot d \cdot x}\right)
\]
is the arc length of the perimeter of circle $\ncalV$ inside circle $\ncalU$, $(g)$ follows from rewriting the two outer expectations as integrals, $(h)$ follows from multiplying by 2 and halving the integration limits of $\theta$ due to symmetry, $(i)$ follows from \eqref{Eq:AreaA}-\eqref{Eq:AreaBnotAintB} where $R_{\ell} = r$ and $\bRL = x$ here, and $(j)$ follows from substituting in the expression of $f_{R_{\ell}}(r;\ell,\pppai)$ provided in \eqref{Eq:P2:R_ell}.
%!TEX root = paper_2.tex

\bibliographystyle{IEEEtran}
\bibliography{paper_2}

% Generated by IEEEtran.bst, version: 1.13 (2008/09/30)
\begin{thebibliography}{10}
\providecommand{\url}[1]{#1}
\csname url@samestyle\endcsname
\providecommand{\newblock}{\relax}
\providecommand{\bibinfo}[2]{#2}
\providecommand{\BIBentrySTDinterwordspacing}{\spaceskip=0pt\relax}
\providecommand{\BIBentryALTinterwordstretchfactor}{4}
\providecommand{\BIBentryALTinterwordspacing}{\spaceskip=\fontdimen2\font plus
\BIBentryALTinterwordstretchfactor\fontdimen3\font minus
  \fontdimen4\font\relax}
\providecommand{\BIBforeignlanguage}[2]{{%
\expandafter\ifx\csname l@#1\endcsname\relax
\typeout{** WARNING: IEEEtran.bst: No hyphenation pattern has been}%
\typeout{** loaded for the language `#1'. Using the pattern for}%
\typeout{** the default language instead.}%
\else
\language=\csname l@#1\endcsname
\fi
#2}}
\providecommand{\BIBdecl}{\relax}
\BIBdecl

\bibitem{Schloemann2015e}
J.~Schloemann, H.~S. Dhillon, and R.~M. Buehrer, ``{Effect of collaboration on
  localizability in range-based localization systems},'' submitted to IEEE
  GLOBECOM 2015 Workshop on \emph{Localization for Indoors, Outdoors, and
  Emerging Networks (LION)}, San Diego, CA, 2015.

\bibitem{Karp2000}
B.~Karp and H.~T. Kung, ``{GPSR: greedy perimeter stateless routing for
  wireless networks},'' in \emph{Proc. ACM Int. Conf. Mob. Comput. Netw.}, New
  York, New York, USA, Aug. 2000, pp. 243--254.

\bibitem{Ko2000}
Y.~Ko and N.~H. Vaidya, ``{Location-aided routing (LAR) in mobile ad-hoc
  networks},'' \emph{Wirel. Networks}, vol.~6, no.~4, pp. 307--321, 2000.

\bibitem{Blazevic2001}
L.~Blazevic, S.~Giordano, and J.-Y. {Le Boudec}, ``{Self-organized routing in
  wide area mobile ad-hoc networks},'' in \emph{Proc. IEEE Glob. Telecommun.
  Conf.}, vol.~5, 2001, pp. 2814--2818.

\bibitem{Jain2001}
R.~Jain, A.~Puri, and R.~Sengupta, ``{Geographical routing using partial
  information for wireless ad-hoc networks},'' \emph{IEEE Pers. Commun.},
  vol.~8, no.~1, pp. 48--57, 2001.

\bibitem{Lee2010}
J.~Lee, S.~Yoo, and S.~Kim, ``{Energy-aware routing in location-based ad-hoc
  networks},'' in \emph{Proc. IEEE Int. Symp. Commun. Control Signal Process.},
  Mar. 2010.

\bibitem{Raji2014}
V.~Raji and N.~M. Kumar, ``{An effective stateless QoS routing for multimedia
  applications in MANET},'' \emph{Int. J. Wirel. Mob. Comput.}, vol.~7, no.~5,
  Sep. 2014.

\bibitem{Akyildiz2002}
I.~Akyildiz, W.~Su, Y.~Sankarasubramaniam, and E.~Cayirci, ``{Wireless sensor
  networks: a survey},'' \emph{Comput. Networks J.}, vol.~38, no.~4, pp.
  393--422, Mar. 2002.

\bibitem{Patwari2005}
N.~Patwari and J.~N. Ash, ``{Locating the nodes: cooperative localization in
  wireless sensor networks},'' \emph{IEEE Signal Process. Mag.}, pp. 54--69,
  Jul. 2005.

\bibitem{Gezici2005}
S.~Gezici, G.~Giannakis, H.~Kobayashi, A.~Molisch, H.~V. Poor, and
  Z.~Sahinoglu, ``{Localization via ultra-wideband radios: a look at
  positioning aspects for future sensor networks},'' \emph{IEEE Signal Process.
  Mag.}, vol.~22, no.~4, pp. 70--84, Jul. 2005.

\bibitem{Sayed2005}
A.~H. Sayed, A.~Tarighat, and N.~Khajehnouri, ``{Network-based wireless
  location: challenges faced in developing techniques for accurate wireless
  location information},'' \emph{IEEE Signal Process. Mag.}, vol.~22, no.~4,
  pp. 24--40, Jul. 2005.

\bibitem{Gustafsson2005}
F.~Gustafsson and F.~Gunnarsson, ``{Mobile positioning using wireless networks:
  possibilities and fundamental limitations based on available wireless network
  measurements},'' \emph{IEEE Signal Process. Mag.}, vol.~22, no.~4, pp.
  41--53, Jul. 2005.

\bibitem{FCCE911CFR}
{Code of Federal Regulations}, ``{911 Service},'' 47 C.F.R. 20.18(h)(2)(ii),
  2015.

\bibitem{FCC2015}
{Federal Communications Commission}, ``{Wireless E911 location accuracy
  requirements},'' PS Docket No. 07-114, Jan. 2015.

\bibitem{R1-091912}
{Third Generation Partnership Project (3GPP)}, ``{R1-091912: Discussions on UE
  positioning issues},'' Nortel, 3GPP TSG-RAN WG1 \#57, San Francisco, USA, May
  2009.

\bibitem{Savvides2001}
A.~Savvides, C.-C. Han, and M.~B. Strivastava, ``{Dynamic fine-grained
  localization in ad-hoc networks of sensors},'' in \emph{Proc. ACM Int. Conf.
  Mob. Comput. Netw.}, New York, New York, USA, Jul. 2001, pp. 166--179.

\bibitem{Alsindi2009a}
N.~Alsindi and C.~Duan, ``{NLOS channel identification and mitigation in
  ultra-wideband TOA-based wireless sensor networks},'' in \emph{Proc. Work.
  Positioning, Navig. Commun.}, Mar. 2009, pp. 59--66.

\bibitem{Shen2010b}
Y.~Shen, H.~Wymeersch, and M.~Z. Win, ``{Fundamental limits of wideband
  localization--part II: cooperative networks},'' \emph{IEEE Trans. Inf.
  Theory}, vol.~56, no.~10, pp. 4981--5000, Oct. 2010.

\bibitem{Vaghefi2014a}
R.~M. Vaghefi and R.~M. Buehrer, ``{Improving positioning in LTE through
  collaboration},'' in \emph{Proc. Work. Positioning, Navig. Commun.}, Mar.
  2014.

\bibitem{Wymeersch2009}
H.~Wymeersch, J.~Lien, and M.~Z. Win, ``{Cooperative localization in wireless
  networks},'' \emph{Proc. IEEE}, vol.~97, no.~2, pp. 427--450, Feb. 2009.

\bibitem{Schloemann2015b}
J.~Schloemann and R.~M. Buehrer, ``{On the value of collaboration in location
  estimation},'' \emph{IEEE Trans. Veh. Tech}, to appear.

\bibitem{Dousse2002}
O.~Dousse, P.~Thiran, and M.~Hasler, ``{Connectivity in ad-hoc and hybrid
  networks},'' in \emph{Proc. IEEE Int. Conf. Comput. Commun.}, vol.~2, 2002,
  pp. 1079--1088.

\bibitem{Bettstetter2005}
C.~Bettstetter and C.~Hartmann, ``{Connectivity of wireless multihop networks
  in a shadow fading environment},'' \emph{Wirel. Networks}, vol.~11, no.~5,
  pp. 571--579, 2005.

\bibitem{Santi2005}
P.~Santi, ``{The critical transmitting range for connectivity in mobile ad-hoc
  networks},'' \emph{IEEE Trans. Mob. Comput.}, vol.~4, no.~3, pp. 310--317,
  May 2005.

\bibitem{Ren2011}
W.~Ren, Q.~Zhao, and A.~Swami, ``{Connectivity of heterogeneous wireless
  networks},'' \emph{IEEE Trans. Inf. Theory}, vol.~57, no.~7, pp. 4315--4332,
  Jul. 2011.

\bibitem{Rekleitis2002}
I.~M. Rekleitis, G.~Dudek, and E.~E. Milios, ``{Multi-robot cooperative
  localization: a study of trade-offs between efficiency and accuracy},'' in
  \emph{Int. Conf. Intell. Robot. Syst.}, vol.~3, 2002, pp. 2690--2695.

\bibitem{Ihler2005}
A.~T. Ihler, J.~W. {Fischer III}, R.~L. Moses, and A.~S. Willsky,
  ``{Nonparametric belief propagation for self-localization of sensor
  networks},'' \emph{IEEE J. Sel. Areas Commun.}, vol.~23, no.~4, pp. 809--819,
  2005.

\bibitem{Alsindi2006}
N.~Alsindi, K.~Pahlavan, B.~Alavi, and X.~Li, ``{A novel cooperative
  localization algorithm for indoor sensor networks},'' in \emph{IEEE Int.
  Symp. Pers. Indoor Mob. Radio Commun.}, Sep. 2006.

\bibitem{Wymeersch2008}
H.~Wymeersch, U.~Ferner, and M.~Z. Win, ``{Cooperative Bayesian self-tracking
  for wireless networks},'' \emph{IEEE Commun. Lett.}, vol.~12, no.~7, pp.
  505--507, Jul. 2008.

\bibitem{Eren2004}
T.~Eren, O.~Goldenberg, W.~Whiteley, Y.~Yang, A.~Morse, B.~Anderson, and
  P.~Belhumeur, ``{Rigidity, computation, and randomization in network
  localization},'' in \emph{Proc. IEEE Int. Conf. Comput. Commun.}, vol.~4,
  2004, pp. 2673--2684.

\bibitem{Andrews2011}
J.~G. Andrews, F.~Baccelli, and R.~K. Ganti, ``{A tractable approach to
  coverage and rate in cellular networks},'' \emph{IEEE Trans. Commun.},
  vol.~59, no.~11, pp. 3122--3134, Nov. 2011.

\bibitem{Dhillon2012}
H.~S. Dhillon, R.~K. Ganti, F.~Baccelli, and J.~G. Andrews, ``{Modeling and
  analysis of K-tier downlink heterogeneous cellular networks},'' \emph{IEEE J.
  Sel. Areas Commun.}, vol.~30, no.~3, pp. 550--560, Apr. 2012.

\bibitem{Schloemann2015c}
J.~Schloemann, H.~S. Dhillon, and R.~M. Buehrer, ``{Towards a tractable
  analysis of localization fundamentals in cellular networks},''
  \emph{arXiv:1502.06899 [cs.IT]}.

\bibitem{Novlan2013}
T.~D. Novlan, H.~S. Dhillon, and J.~G. Andrews, ``{Analytical modeling of
  uplink cellular networks},'' \emph{IEEE Trans. Wirel. Commun.}, vol.~12,
  no.~6, pp. 2669--2679, Jun. 2013.

\bibitem{Stoyan1995}
D.~Stoyan, W.~S. Kendall, and J.~Mecke, \emph{{Stochastic Geometry and Its
  Applications}}, 2nd~ed.\hskip 1em plus 0.5em minus 0.4em\relax Chichester:
  John Wiley and Sons, 1995.

\bibitem{Haenggi2013}
M.~Haenggi, \emph{{Stochastic Geometry for Wireless Networks}}.\hskip 1em plus
  0.5em minus 0.4em\relax New York: Cambridge University Press, 2013.

\bibitem{R1-091443}
{Third Generation Partnership Project (3GPP)}, ``{R1-091443: Evaluation
  parameters for positioning studies},'' Alcatel-Lucent, Ericsson, Motorola,
  Nokia, NSN, Nortel, Qualcomm Europe, 3GPP TSG-RAN WG1 \#56bis, Seoul, Korea,
  Mar. 2009.

\bibitem{Daneshgaran2007}
F.~Daneshgaran, M.~Laddomada, and M.~Mondin, ``{Connection between system
  parameters and localization probability in network of randomly distributed
  nodes},'' \emph{IEEE Trans. Wirel. Commun.}, vol.~6, no.~12, pp. 4383--4389,
  Dec. 2007.

\bibitem{Fischer2014}
S.~Fischer, ``{Observed Time Difference Of Arrival (OTDOA) positioning in 3GPP
  LTE},'' \emph{Qualcomm White Pap.}, 2014.

\bibitem{Goldenberg2005}
D.~K. Goldenberg, A.~Krishnamurthy, W.~C. Maness, Y.~R. Yang, A.~Young, A.~S.
  Morse, A.~Savvides, and B.~D.~O. Anderson, ``{Network localization in
  partially localizable networks},'' in \emph{Proc. IEEE Int. Conf. Comput.
  Commun.}, vol.~1, 2005, pp. 313--326.

\bibitem{Yang2012}
Z.~Yang and Y.~Liu, ``{Understanding node localizability of wireless ad-hoc and
  sensor networks},'' \emph{IEEE Trans. Mob. Comput.}, vol.~11, no.~8, pp.
  1249--1260, Aug. 2012.

\bibitem{Laman2002}
G.~Laman, ``{On graphs and rigidity of plane skeletal structures.}'' \emph{J.
  Eng. Math.}, vol.~4, pp. 331--340, 2002.

\bibitem{Buehrer12}
R.~M. Buehrer and S.~Venkatesh, ``{Fundamentals of time-of-arrival-based
  position location},'' in \emph{Handbook of Position Location: Theory,
  Practice, and Advances}, S.~A. Zekavat and R.~M. Buehrer, Eds.\hskip 1em plus
  0.5em minus 0.4em\relax Hoboken, NJ: IEEE Press/John Wiley and Sons, 2012.

\bibitem{Keeler2013}
H.~P. Keeler, B.~Blaszczyszyn, and M.~K. Karray, ``{SINR-based k-coverage
  probability in cellular networks with arbitrary shadowing},'' in \emph{IEEE
  Int. Symp. Inf. Theory}, Jul. 2013, pp. 1167--1171.

\bibitem{LuneArea}
E.~W. Weisstein, ``{Lune},'' From MathWorld--A Wolfram Web Resource.
  http://mathworld.wolfram.com/Lune.html.

\bibitem{Haenggi2005}
M.~Haenggi, ``{On distances in uniformly random networks},'' \emph{IEEE Trans.
  Inf. Theory}, vol.~51, no.~10, pp. 3584--3586, Oct. 2005.

\bibitem{3GPP.36.305}
{Third Generation Partnership Project (3GPP)}, ``{Evolved Universal Terrestrial
  Radio Access Network (E-UTRAN); Stage 2 functional specification of User
  Equipment (UE) positioning in E-UTRAN},'' Mar. 2013.

\end{thebibliography}

\end{document}